\documentclass[preprint,11pt]{article}
\usepackage[affil-it]{authblk}
\usepackage[left=2cm,right=2cm]{geometry}

\newcommand{\lv}[1]{}

\newcommand{\veps}{\varepsilon}

\newcommand{\X}{\mathcal{X}}

\usepackage{multirow}

\usepackage{array, makecell, cellspace}
\usepackage{setspace}
\usepackage{adjustbox}

\makeatletter
\newcommand{\vast}{\bBigg@{3.5}}
\makeatother

\usepackage{floatpag}
\usepackage{afterpage}

\usepackage[bottom]{footmisc}

\usepackage{enumitem}

\usepackage{mdframed}

\usepackage{float}

\usepackage[linesnumbered,lined]{algorithm2e}

\SetCommentSty{mycommfont}
\DontPrintSemicolon
\usepackage{footnote}
\usepackage{tablefootnote}
\makesavenoteenv{tabular}

\usepackage{newfloat}
\DeclareFloatingEnvironment[fileext=lop]{Algorithm}

\usepackage[labelfont=bf]{caption}

\usepackage{geometry}
\usepackage{framed}
\usepackage{float}
\restylefloat{figure}

\usepackage{lineno}
\usepackage{amssymb}
\usepackage{mathrsfs}
\usepackage{amsmath}
\usepackage{amsthm}
\usepackage{dsfont}
\usepackage{bbm}
\usepackage{mathtools}
\usepackage{hyperref}
\newtheorem{definition}{Definition}[section]
\newtheorem{theorem}{Theorem}
\newtheorem{corollary}{Corollary}[theorem]
\newtheorem{lemma}[theorem]{Lemma}

\newtheorem{prop}{Proposition}

\usepackage[table]{xcolor}
\usepackage{array, makecell}

\usepackage[normalem]{ulem}

\usepackage{graphicx}
\usepackage{relsize}
\usepackage[numbers]{natbib}

\allowdisplaybreaks

\usepackage{xcolor}
\usepackage{thmtools,thm-restate}

\definecolor{mycolor}{rgb}{1, 0.0, 0.0}

\newcommand{\costt}[1]{\textup{cost($#1$)}}
\newcommand{\ftcostt}[1]{\textup{$\Phi$($#1$)}}
\newcommand{\rp}[2]{#1\strut^{\hspace{0.15mm}#2}}

\usepackage{scalerel}[2016/12/29]

\newcommand{\zl}{\scaleobj{1.15}{z}}
\newcommand{\OC}{\mathbb{O}}

\allowdisplaybreaks

\graphicspath{{./Figures/}}

\begin{document}


\title{\vspace{-2.5cm}Tight FPT Approximation for Constrained k-Center and k-Supplier}
%\lv{
%\subtitle{(Full Version)}
%}
%
%\titlerunning{}
% abbreviated title (for running head)
\author{Dishant Goyal and Ragesh Jaiswal \\ Department of Computer Science and Engineering, \\
Indian Institute of Technology Delhi  \thanks{Email addresses: \texttt{\{Dishant.Goyal, rjaiswal\}@cse.iitd.ac.in}} }
\date{}
%
%\authorrunning{Jaiswal et al.}
% abbreviated author list (for running head)
%\tocauthor{Ragesh Jaiswal}
%
{\def\addcontentsline#1#2{}\maketitle}
%\maketitle     % typeset the title of the contribution
%%\thispagestyle{empty}

\begin{abstract}
In this work, we study a range of \emph{constrained} versions of the \emph{$k$-supplier} and \emph{$k$-center} problems. In the classical (\emph{unconstrained}) $k$-supplier problem, we are given a set of clients $C$ in a metric space $\mathcal{X}$, with distance function $d(.,.)$. We are also given a set of feasible facility locations $L \subseteq \mathcal{X}$. The goal is to open a set $F$ of $k$ facilities in $L$ to minimize the maximum distance of any client to the closest open facility, i.e., minimize, $\costt{F,C} \equiv \max_{j \in C} \big\{ d(F,j) \big\}$, where $d(F,j)$ is the distance of client $j$ to the closest facility in $F$. The $k$-center problem is a special case of the $k$-supplier problem where $L = C$.

In this work, we study various constrained versions of the $k$-supplier problem such as: \emph{capacitated}, \emph{fault-tolerant}, $\ell$-diversity, etc. These problems fall under a broad framework of \emph{constrained clustering}. 
A unified framework for constrained clustering was proposed by Ding and Xu~\cite{constrained:2015_Ding_and_Xu} in context of the $k$-median and $k$-means objectives. In this work, we extend this framework to the $k$-supplier and $k$-center objectives. 
This unified framework allows us to obtain results simultaneously for the following constrained versions of the $k$-supplier problem: $r$-gather, $r$-capacity, balanced, chromatic, fault-tolerant, strongly private, $\ell$-diversity, and fair $k$-supplier problems, with and without outliers.
We obtain the following results:
\begin{itemize}
\item We give $3$ and $2$ approximation algorithms for the constrained $k$-supplier and $k$-center problems, respectively, with $\mathsf{FPT}$ running time $k^{O(k)} \cdot n^{O(1)}$, where $n = |C \cup L|$. Moreover, these approximation guarantees are tight; that is, for any constant $\veps>0$, no algorithm can achieve $(3-\veps)$ and $(2-\veps)$ approximation guarantees for the constrained $k$-supplier and $k$-center problems in $\mathsf{FPT}$ time, assuming $\mathsf{FPT} \neq \mathsf{W}[2]$. \item We study the constrained clustering problem with outliers. Our algorithm gives $3$ and $2$ approximation guarantees for the constrained \emph{outlier} $k$-supplier and $k$-center problems, respectively, with $\mathsf{FPT}$ running time $(k+m)^{O(k)} \cdot n^{O(1)}$, where $n = |C \cup L|$ and $m$ is the number of outliers.    %Moreoever, our algorithm works for the \emph{weighted} version of the constrained clustering problem, where every client $j \in C$ has a positive real weight $w(j)$. 
\item Our techniques generalise for distance function $d(., .)^z$. That is, for any positive real number $z$, if the cost of a client is defined by $d(.,.)^{\zl}$ instead of $d(.,.)$, then our algorithm gives $\rp{3}{\zl}$ and $\rp{2}{\zl}$ approximation guarantees for the constrained $k$-supplier and $k$-center problems, respectively. 
%In particular, we study the following constrained versions of the $k$-supplier problem: $r$-gather, $r$-capacity, balanced, chromatic, fault-tolerant, strongly private, $\ell$-diversity, and fair $k$-supplier problems, with and without outliers. 
\end{itemize}
\end{abstract}

\section{Introduction}\label{section:introduction}
We start with the definition of the \emph{$k$-supplier} problem:

\begin{definition}[$k$-Supplier]
Let $(\X, d)$ be a metric space. Let $k$ be any positive integer and $z$ be any positive real number. Given a set $L \subseteq \X$ of feasible facility locations, and a set $C \subseteq \X$ of clients, find a set $F \subseteq L$ of $k$ facilities that minimises the cost: $\costt{F,C} \equiv \max_{j \in C} \Big\{ \min_{i \in F} \big\{ d(i,j)^{\zl} \big\} \Big\}$.
\end{definition}

\noindent When $L = C$, the problem is known as the \emph{$k$-center} problem. In the above definition, we did not impose any constraint on the clients. Therefore, this version is also known as the~\emph{unconstrained} $k$-supplier problem. Note that we will use the term $k$-supplier problem and~\emph{unconstrained} $k$-supplier problem interchangeable in this paper. The $k$-supplier and $k$-center problems have natural applications in deciding the optimal location of placing the facilities like: hospitals, schools, post offices, etc. in a geographical area~\cite{survey:2004_facility_location_Daskin,survey:2017_facility_location_Amir}. It ensures that no client pays a very high transportation cost for availing a particular facility. 
Furthermore, these problems are extensively used for clustering large data sets in areas such as data mining, pattern recognition, information retrieval, etc.. The clients that are assigned to the same facility belongs to the same cluster, and the corresponding facility is known as the \emph{cluster center}. Keeping this in mind, from now on, we will use the term \emph{facility} and \emph{center} interchangeably. 

In many applications, additional constraints are imposed on the clusters. For example, in privacy preserving clustering, every cluster is required to have at least $r$ clients. This problem is known as the \emph{$r$-gather $k$-supplier}/\emph{center} problem or the \emph{lower bounded $k$-supplier}/\emph{center} problem~\cite{rgather:k_anonymity_2005_Aggarwal,rgather:k_anonymity_2002_Sweeney}. 
Similarly, in many resource allocation problems, we have an upper bound constraint imposed on every facility location. That is, a facility can not serve more than $u$ clients, for some integer constant $u>0$. This ensures that the load is almost equally distributed among the facilities. This problem is known as the \emph{capacitated $k$-supplier}/\emph{center} problem~\cite{capacitated:kcenter_1992_Barillan,capacitated:kcenter_2000_khuller,capacitated:kcenter_2012_khuller,capacitated:kcenter_2015_An}. Likewise, there are many other constrained versions of the $k$-supplier/center problems namely \emph{fault-tolerant}~\cite{fault:kcenter_2000_khuller}, \emph{fair}~\cite{fairness:2019_Bera_NIPS}, \emph{chromatic}~\cite{constrained:2015_Ding_and_Xu}, \emph{$\ell$-diversity}~\cite{L_diversity:2010_kcenter_Li_Jian} problems, etc..

In the past, many of the constrained versions of the clustering problems were studied separately as independent problems. 
Recently, in $2015$, Ding and Xu~\cite{constrained:2015_Ding_and_Xu} gave a unified framework for these problems that they called the \emph{constrained clustering} framework. They proposed this unified framework in the context of the $k$-median and $k$-means problems in the continuous Euclidean spaces. 
More recently, Goyal~\emph{et al.}~\cite{constrained:2020_GJK_FPT} extended the framework to general discrete metric spaces. Similar results were obtained by Feng~\emph{et al.}~\cite{fpt:2020_Feng_Zhang_Unified_Framework} using a slightly different framework. In this work, we further extend the unified framework to the $k$-supplier and $k$-center objectives. Using this, we design $\mathsf{FPT}$ time approximation algorithms for a range of constrained clustering problems. To put it another way, the unified framework allows us to use a common algorithmic technique for various constrained clustering problems. A natural question we need to address at the beginning of this work is: ``\textit{why should one be interested in designing approximation algorithms with $\mathsf{FPT}$ running time for these problems?}" 
The answer lies in the fact that the $k$-supplier and $k$-center problems are $\mathsf{W}[2]$-hard parameterized by $k$~\cite{hardness:2005_kcenter_W2hard_Demaine}. Therefore, we can not obtain exact algorithms for the $k$-supplier and $k$-center problems in $\mathsf{FPT}$ running time unless $\mathsf{W}[2] = \mathsf{FPT}$. Very recently, stronger $\mathsf{FPT}$ hardness results have been established for these problems. The following are these two results that easily follow from the work of Hochbaum and Shmoys~\cite{Supplier:1986_Shmoys}, and Pasin Manurangsi~\cite{fpt:2020_Pasin_SODA}.
\begin{theorem}\label{theorem:hardness_ksupplier}
For any constant $\veps>0$, $z >0$, and any function $g \colon \mathbb{Z}^{+} \to \mathbb{R}^{+}$, the $k$-supplier problem can not be approximated to factor $(\rp{3}{z} - \veps)$ in time $g(k) \cdot n^{o(k)}$, assuming Gap-ETH.
\end{theorem}

\begin{theorem}\label{theorem:hardness_kcenter}
For any constant $\veps>0$, $z >0$, and any function $g \colon \mathbb{Z}^{+} \to \mathbb{R}^{+}$, the $k$-center problem can not be approximated to factor $(\rp{2}{z} - \veps)$ in time $g(k) \cdot n^{o(k)}$, assuming Gap-ETH.
\end{theorem}

\noindent A detailed proof of Theorem~\ref{theorem:hardness_ksupplier} is given Appendix D of~\cite{fairness:2021_GJ}. For Theorem~\ref{theorem:hardness_kcenter}, we give a detailed proof in Appendix~\ref{appendix:reduction} of this paper. 
Note that the above hardness results applies for all the constrained clustering problems that we study in this paper. 
Our main contribution is to give matching upper bounds for all the constrained clustering problems. That is, we design $\mathsf{FPT}$ time $\rp{3}{z}$ and $\rp{2}{z}$ approximation algorithms for various constrained $k$-supplier and $k$-center problems, respectively. 

Let us now define the constrained clustering framework for the $k$-supplier and $k$-center problems. 
Let $\mathbb{O} = \{O_{1},\dotsc,O_k\}$ be any arbitrary partitioning of the client set $C$. 
Let $F \subseteq L$ be any set of $k$ facilities. 
Let $f_{i}^{*}$ be a facility in $F$ that minimizes the $1$-supplier cost of partition $O_i$. 
That is, $f_{i}^{*}$ is the facility in $F$ that minimises $\max_{x \in O_i} \{d(c, f_{i}^{*})^z\}$.
Then, the $k$-supplier cost of the partitioning $\mathbb{O}$ with respect to the facility set $F$ is given as follows:

$$\Psi(F, \mathbb{O}) \equiv \max_{i = 1}^{k} \Big\{ \max_{x \in O_{i}} \big\{ d(x, f_{i}^{*})^{
\zl} \big\} \Big\}$$

\noindent In other words, a partition $O_i$ is completely assigned to a facility location $f_{i}^{*}$ in $F$, and the assignment cost of every client in $O_i$ is measured with respect to $f_{i}^{*}$. 
Then, $\Psi(F,\mathbb{O})$ is simply the maximum assignment cost over all the clients. Furthermore, the optimal $k$-supplier cost of $\mathbb{O}$ is given as follows: $\Psi^{*}(\mathbb{O}) \equiv \min\limits_{\textrm{$k$-center-set $F$}} \Psi(F,\mathbb{O})$. Now, suppose that we are given a collection $\mathbb{S} = \{\mathbb{O}_{1},\dotsc,\mathbb{O}_{t} \}$ of $t$ different partitionings of $C$. The goal of the constrained clustering problem is to find a partitioning in $\mathbb{S}$ that has the minimum $k$-supplier cost.
Formally, we define the problem as follows: 
\begin{definition}[Constrained $k$-Supplier Problem]\label{definition:constrained_k_supplier}
Let $(\X, d)$ be a metric space, $k$ be any positive integer, and $z$ be any positive real number.
Given a set $L \subseteq \X$ of feasible facility locations, a set $C \subseteq \X$ of clients, and a set $\mathbb{S}$ of feasible partitionings of $C$, find a partitioning $\mathbb{O} = \{O_{1}, O_{2}, \dotsc, O_{k}\}$ in $\mathbb{S}$, that minimizes the cost function:
$\Psi^{*}(\mathbb{O}) ~ \equiv ~ \min\limits_{\textrm{$k$-center-set $F$}} \Psi(F, \OC)$.
\end{definition}
\noindent The above definition encapsulates many constrained clustering problems. For example, consider the $r$-gather clustering problem, in which the goal is to find a clustering $\mathbb{O} = \{O_1,\dotsc,O_k\}$ of the client set such that the $i^{th}$ cluster has at least $r_i$ clients in it, for some constant $r_i \geq 0$. For this problem, the set $\mathbb{S}$ can be concisely defined as $\mathbb{S} \coloneqq \{ \OC ~ \mid ~ \textrm{for every cluster } O_{i} \in \OC$, $|O_{i}| \geq r_i \}$, where $\OC = \{O_{1},O_{2},\dotsc,O_{k} \}$ is a partitioning of the client set. Table~\ref{table:definitions} describes a list of other constrained clustering problems that satisfy the above definition of constrained $k$-supplier problem. In the list, one exception is the \emph{fault-tolerant $k$-supplier} problem; it does not satisfy the definition of the constrained $k$-supplier problem since its objective function is different from $\Psi(F,\mathbb{O})$. In Appendix~\ref{appendix:fault_tolerant}, we show that the fault-tolerant $k$-supplier problem can be reduced to the \emph{chromatic $k$-supplier} problem, which satisfies the definition of the constrained $k$-supplier problem. We would also like to point out that we are considering the \emph{soft} assignment version of the constrained clustering problems. That is, it is allowed to open more than one facility at any particular location in $L$. 
This version is different from the \emph{hard} assignment version where a single copy of a facility can be opened at any particular location in $L$. Note that in both the versions the total number of open facilities are at most $k$. 
Note that the soft assignment version is easier than the hard assignment version since the soft assignment version can be reduced to the hard assignment version by creating $k$ copies of every location in $L$. 

\begin{table}[h]
\begin{adjustbox}{width=\columnwidth,center}
\centering
\setcellgapes{1ex}\makegapedcells
\begin{tabular}{|l|l|l|}
\hline
\# & {\bf Problem} & {\bf Description} \\ \hline
1. & \makecell[l]{$r$-Gather $k$-supplier problem } & \makecell[l]{Given $k$ positive integers: $r_1,\dotsc,r_k$, find clustering $\OC = \{O_1, ..., O_k\}$ \\ with minimum $\Psi^{*}(\OC)$ such that for all $i$, $|O_i| \geq r_{i}$}.\\ \hline
2. & \makecell[l]{$r$-Capacity $k$-supplier problem } & \makecell[l]{Given $k$ positive integers: $r_1,\dotsc,r_k$, find clustering $\OC = \{O_1, ..., O_k\}$ \\ with minimum $\Psi^{*}(\OC)$ such that for all $i$, $|O_i| \leq r_{i}$} \\ \hline
3. & \makecell[l]{Balanced $k$-supplier problem } & \makecell[l]{Given positive integers: $\ell_1,\dotsc,\ell_k$, and $r_1,\dotsc,r_k$, find clustering \\ $\OC = \{O_1, ..., O_k\}$ with minimum $\Psi^{*}(\OC)$ such that for all $i$, $\ell_{i} \leq |O_i| \leq r_{i}$} \\ \hline
4. & \makecell[l]{Chromatic $k$-supplier problem } & \makecell[l]{Given that every client has an associated color,  find a clustering \\ $\OC = \{O_1, ..., O_k\}$ with minimum $\Psi^{*}(\OC)$ such that for all $i$, $O_i$ should \\ not have any two points with the same color.} \\ \hline
5. & \makecell[l]{Fault-tolerant $k$-supplier problem } & \makecell[l]{Given positive integer $l_{x} \leq k$ for every client $x \in C$, \\ find a set $F$ of $k$ centers, such that the maximum assignment cost \\ of $x$ to $l_{x}^{th}$ closest facility is minimized.}  \\ \hline
%
% 8. & \makecell[l]{Balanced Fair $k$-supplier problem \\ $k$-{\tt BFsupplier}} & \makecell[l]{Given the sets of positive integers: $\{\ell_{1},\dotsc,\ell_{k}\}$ and $\{r_1,\dotsc,r_k\}$; \\ a partitioning $C_1,\dotsc,C_{t}$ of the client set $C$, and two fairness vectors \\ $\alpha,\beta \in [0,1]^{t}$, find a clustering $\mathbb{O} = \{O_1,\dotsc,O_k\}$ with minimum $\Psi^{*}(\mathbb{O})$ \\ such that it satisfies that $\beta_j \cdot |O_i| \leq |O_i \cap C_j| \leq \alpha_j \cdot |O_i|$ and $\ell_{i} \leq |C_i| \leq r_i$ \\ for every $i \in [k]$ and $j \in [t]$.} \\ \hline
% %
6. & \makecell[l]{Strongly private $k$-supplier problem} & \makecell[l]{Given a partitioning $C_1,\dotsc,C_{\omega}$ of the client set $C$, and a set of integers: \\ $\{\ell_1,\dotsc,\ell_{\omega}\}$,  find a clustering $\mathbb{O} = \{O_1, ..., O_{k}\}$ with minimum $\Psi^{*}(\mathbb{O})$ that \\ satisfies $|C_j \cap O_i| \geq \ell_{j}$ for every $i \in [k]$ and $j \in [\omega]$.} \\ \hline
7. & \makecell[l]{$\ell$-Diversity $k$-supplier problem } & \makecell[l]{Given a partitioning $C_1,\dotsc,C_{\omega}$ of the client set $C$, a real number $\ell>1$,\\find a clustering $\OC = \{O_1, ..., O_k\}$ with minimum $\Psi^{*}(\OC)$ such that \\ the fraction of points belonging to the same partition inside $O_i$ is $\leq 1/\ell$.} \\ \hline
8. & \makecell[l]{Fair $k$-supplier problem} & \makecell[l]{Given $\omega$ color classes $C_1,\dotsc,C_{\omega}$ (not necessarily disjoint), such that every \\ $C_j$ is a subset of the client set $C$, and two fairness vectors $\alpha,\beta$ $\in [0,1]^{\omega}$, \\ find a clustering $\mathbb{O} = \{O_1,\dotsc,O_k\}$ with minimum $\Psi^{*}(\mathbb{O})$  such that it \\ satisfies that $\beta_j \cdot |O_i| \leq |O_i \cap C_j| \leq \alpha_j \cdot |O_i|$  for every $i \in [k]$ and $j \in [\omega]$.} \\ \hline
%
% 8. & \makecell[l]{Colorful $k$-supplier problem \\ $k$-{\tt Cosupplier}} & \makecell[l]{The client set $C$ is partitioned into $\omega$ color classes $C_1,\dotsc,C_{\omega}$ with \\ coverage requirements $p_1,\dotsc,p_{\omega}$. \\, find a center set $F = \{f_1,\dotsc,f_k\}$ with minimum $\Psi^{*}(\mathbb{O})$ that \\ satisfies $|C_j \cap O_i| \geq \ell_{j}$ for every $i \in [k]$ and $j \in [t]$.} \\ \hline
% %
% 9. & \makecell[l]{Non-Uniform $k$-supplier problem \\ $k$-{\tt NUsupplier}} & \makecell[l]{Given $k$ radii $r_1,\dotsc,r_k$, find a clustering $\mathbb{O} = \{O_1, ..., O_{k}\}$ and a center set \\ $F = \{f_1,\dotsc,f_k\}$ that satisfies $\max_{x \in O_i} \big\{ d(x,f_i) \big\} \leq r_i$ for every $i \in [k]$.} \\ \hline
%
% 11. & \makecell[l]{Outlier $k$-supplier problem \\ $k$-{\tt Osupplier}} & \makecell[l]{Find a subset $Z\subseteq C$ of size $m$ and a clustering $\mathbb{O}' = \{O_1', ..., O_{k}'\}$ of the \\ set $C' \coloneqq C \setminus Z$, such that $\Psi^{*}(\mathbb{O}')$ is minimized.} \\ \hline
% %
% 12. & \makecell[l]{Balanced Outlier $k$-supplier problem \\ $k$-{\tt BOsupplier}} & \makecell[l]{Given the sets of positive integers: $\{\ell_{1},\dotsc,\ell_{k}\}$ and $\{r_1,\dotsc,r_k\}$, \\ find a subset $Z\subseteq C$ of size $m$ and a clustering $\mathbb{O}' = \{O_1', ..., O_{k}'\}$ of the \\ set $C' \coloneqq C \setminus Z$, with minimum $\Psi^{*}(\mathbb{O}')$ such that it satisfies $\ell_{i} \leq |O'_{i}| \leq r_i$} \\ \hline
%
\end{tabular}
\end{adjustbox}
\caption{Constrained $k$-supplier problems with $\mathsf{FPT}$ time partition algorithms (see Appendix~\ref{appendix:partition_algorithm}).
}\label{table:definitions}
\end{table}

Now, we describe a general algorithmic technique to solve any problem that satisfies the definition of the constrained $k$-supplier problem. More precisely, we show that any constrained $k$-supplier problem can be solved using two basic ingredients: the \emph{list $k$-clustering problem} and a \emph{partition algorithm}. The notion of the list $k$-clustering problem was formalised by Bhattacharya~\emph{et al.}~\cite{constrained:2016_bjk} in the context of the $k$-median and $k$-means objectives. We extend the notion to the $k$-supplier objective as follows:
\begin{definition}[List $k$-Supplier]
Let $\mathcal{I} = (L,C,k,d,z)$ be any instance of the $k$-supplier problem. The goal of the problem is: given $\mathcal{I}$, find a list $\mathcal{L}$ of $k$-center-sets (i.e., each element of the list is a set of $k$ elements from $L$) such that for any partitioning  $\mathbb{O} = \{O_1,\dotsc,O_k\}$ of the client set $C$, the list $\mathcal{L}$ contains a $k$-center-set $F$ such that $\Psi(F,\OC) \leq \alpha \cdot \Psi^{*}(\OC)$ for $\alpha = \rp{3}{z}$. For the $k$-center objective $\alpha = \rp{2}{z}$.
\end{definition}

\noindent Furthermore, we define a partition algorithm as follows:
\begin{definition}[Partition Algorithm]
Let $\mathcal{I} = (L,C,k,d,z)$ be any instance of the $k$-supplier problem, and let $\mathbb{S} = \{\OC_{1}, \OC_{2}, \dotsc,\OC_{t}\}$ be a collection of clusterings of $C$. Given a center set $F \subseteq L$, a partition algorithm outputs a clustering in $\mathbb{S}$ that has the least clustering cost $\Psi(F,\mathbb{O})$ with respect to $F$.
\end{definition}
\noindent Note that the set $\mathbb{S}$ differs for different constrained clustering problems; therefore, the partition algorithm differs for different constrained clustering problems. 
The simplest example of the partition algorithm is for the unconstrained $k$-supplier problem. For the unconstrained $k$-supplier problem, the set $\mathbb{S}$ is the collection of all possible $k$-partitionings of $C$ and the partition algorithm is simply the standard Voronoi partitioning algorithm. 
That is, the algorithm assigns a client to its closest facility location in $F$ and the obtained partitioning has the least clustering cost with respect to $F$.

Now, suppose that we have an algorithm for the list $k$-supplier problem and a partition algorithm for a particular constrained $k$-supplier problem.
Then we can obtain an approximation algorithm for that constrained $k$-supplier problem. 
The following theorem proves this result:
\begin{theorem}\label{theorem:list_alpha_approx}
Let $\mathcal{I} = (L,C,k,d,z,\mathbb{S})$ be any instance of the constrained $k$-supplier problem, and let $A_{\mathbb{S}}$ be a partition algorithm for $\mathbb{S}$ with running time $T_{A}$.
Let $B$ be any algorithm for the list $k$-supplier problem with running time $T_{B}$.
%Suppose we are given a list $\mathcal{L}$, then we can design 
Then, there is an algorithm that outputs a clustering $\OC \in \mathbb{S}$ that is an $\alpha$-approximate solution for the constrained $k$-supplier instance $\mathcal{I}$. For the $k$-supplier objective, $\alpha = \rp{3}{z}$, and for the $k$-center objective, $\alpha = \rp{2}{z}$. Moreover, the running time of the algorithm is $O(T_B + |\mathcal{L}| \cdot T_A)$.
\end{theorem}

\begin{proof}
The proof easily follows from the definitions of the list $k$-supplier problem and the partition algorithm; therefore, we defer the proof to Appendix~\ref{appendix:proof_theorem_3}.
\end{proof}

The goal now becomes to design an algorithm for the list $k$-supplier problem and the partition algorithms for different constrained clustering problems. 
In Section~\ref{section:list_k_supplier}, we design an $\mathsf{FPT}$ time algorithm for the list $k$-supplier problem. 
Formally, we state the result as follows:

\begin{theorem}\label{theorem:list_k_supplier}
Given an instance $\mathcal{I} = (L,C,k,d,z)$ of the $k$-supplier problem. There is an algorithm for the list $k$-supplier problem that outputs a list $\mathcal{L}$ of size at most $k^{O(k)} \cdot n$. Moreover, the running time of the algorithm is $k^{O(k)} \cdot n + O(n^2 \log n)$, which is $\mathsf{FPT}$ in $k$.
\end{theorem}

\noindent Using Theorems~\ref{theorem:list_alpha_approx} and~\ref{theorem:list_k_supplier}, we further obtain the following two corollaries:

\begin{corollary} [Main Result]\label{corollary:1}
For any constrained version of the $k$-supplier problem that has a partition algorithm with running time $T$, there exists a $\rp{3}{z}$ approximation algorithm with running time $T \cdot k^{O(k)} \cdot n + O(n^2 \log n)$.
\end{corollary}

\begin{corollary} [Main Result]\label{corollary:2}
For any constrained version of the $k$-center problem that has a partition algorithm with running time $T$, there exists a $\rp{2}{z}$ approximation algorithm with running time $T \cdot k^{O(k)} \cdot n + O(n^2 \log n)$.
\end{corollary}

The only remaining task is to design partition algorithms for different constrained clustering problems. Note that all the problems described in Table~\ref{table:definitions} satisfy the definition of the constrained $k$-supplier problem. Furthermore, in Appendix~\ref{appendix:partition_algorithm}, we design $\mathsf{FPT}$ time partition algorithms for these problems. Therefore, the above two corollaries imply $\mathsf{FPT}$ time $\rp{3}{z}$ and $\rp{2}{z}$
approximation algorithms for these problems. For the first six problems in Table~\ref{table:definitions}, the running time of the partition algorithms is $k^{k} \cdot n^{O(1)}$. 
For the seventh problem in the table, i.e., the $\ell$-diversity $k$-supplier problem, the running time of the partition algorithm is $(\omega k)^{\omega k} \cdot n^{O(1)}$.  For the last problem in the table, i.e., the fair $k$-supplier problem, the running time of the partition algorithm is $(k \Gamma)^{O(k  \Gamma)} \cdot n^{O(1)}$, where $\Gamma$ is the number of distinct collection of color classes induced by the colors of clients. We define the notation $\Gamma$ more clearly in Appendix~\ref{appendix:fair_partition}.
Thus, we get $\mathsf{FPT}$ time $\rp{3}{z}$ and $\rp{2}{z}$ approximation guarantees for these problems with respect to $k$-supplier and $k$-center objectives, respectively. Formally, we state these results as follows:
\begin{theorem}\label{theorem:1}
There is an $\mathsf{FPT}$ time $\rp{3}{z}$ approximation algorithm for the following constrained versions of the $k$-supplier problem:
\begin{enumerate}
    \item $r$-gather $k$-supplier problem
    \item $r$-capacity $k$-supplier problem
    \item Balanced $k$-supplier problem
    \item Chromatic $k$-supplier problem
    \item Fault-tolerant $k$-supplier problem
    \item Strongly private $k$-supplier problem
\end{enumerate}
The running time of the algorithm is $k^{O(k)} \cdot n^{O(1)}$. Furthermore, for the $k$-center version, the approximation guarantee is $\rp{2}{z}$.
\end{theorem}

\noindent All the above mentioned problems generalize the unconstrained $k$-supplier and $k$-center problems. Therefore, from Theorems~\ref{theorem:hardness_ksupplier} and~\ref{theorem:hardness_kcenter}, it follows that better approximation guarantees are not possible for the above problems in $\mathsf{FPT}$ time assuming Gap-ETH. Thus, it settles the complexity of the above problems, parameterized by $k$. For the $\ell$-diversity and fair $k$-supplier problems, we obtain the following results:

\begin{theorem}\label{theorem:2}
There is $\mathsf{FPT}$ time $\rp{3}{z}$ and $\rp{2}{z}$ approximation algorithms for the $\ell$-diversity $k$-supplier and $k$-center problems, respectively, with running time $(k \omega)^{O(k  \omega)} \cdot n^{O(1)}$.
\end{theorem}

\begin{theorem}\label{theorem:3}
There is $\mathsf{FPT}$ time $\rp{3}{z}$ and $\rp{2}{z}$ approximation algorithms for the fair $k$-supplier and $k$-center problems, respectively, with running time $(k \Gamma)^{O(k  \Gamma)} \cdot n^{O(1)}$, where $\Gamma$ denote the number of distinct collection of color classes induced by the colors of clients. Moreover, if the color classes are pair-wise disjoint, then $\Gamma = \omega$, and the running time of the algorithm is $(k \omega)^{O(k  \omega)} \cdot n^{O(1)}$.
\end{theorem}

\noindent The $\ell$-diversity $k$-supplier problem for $\omega = 1$ and $\ell = 1$, is equivalent to the unconstrained $k$-supplier problem. Also, the fair $k$-supplier problem for $\omega = 1$, $\beta_j = 0$, and $\alpha_j = 1$, is equivalent to the unconstrained $k$-supplier problem. Therefore, better approximation guarantees are not possible for these problems in $\mathsf{FPT}$ time parameterized by $k$ and $\Gamma$ (or $\omega$). The statement simply follows from Theorems~\ref{theorem:hardness_ksupplier} and~\ref{theorem:hardness_kcenter}. In the next subsection, we further extend the constrained clustering framework to the outlier setting. The discussion will be analogous to the above discussion.

\subsection{Constrained Clustering With Outliers}
In practical scenarios, it often happens that a few clients are located at a far away locations from the majority of the clients. These clients are called \emph{outliers}. The presence of outliers forces the algorithm to open the facilities close to the outliers. Due to this, the majority of the clients have to pay high assignment cost. This leads to poor clustering of the dataset. To overcome this issue, we cluster the dataset without the outliers. This gives rise to the \emph{outlier $k$-supplier} problem. 
A mathematical formulation of the problem was given by Charikar~\emph{et al.}~\cite{outlier:kcenter_2001_Charikar} for which they gave a polynomial time $3$-approximation algorithm. The following is the definition of the outlier $k$-supplier problem:

\begin{definition}[Outlier $k$-Supplier]\label{definition:outlier_k_supplier}
Let $(\X, d)$ be a metric space. Let $k$ and $m$ be any positive integers, and $z$ be any positive real number. Given a set $L \subseteq \X$ of feasible facility locations, and a set $C \subseteq \X$ of clients, find a subset $Z \subseteq C$ of size at most $m$ clients and a set $F \subseteq L$ of $k$ facilities such that the $k$-supplier cost of $C' \coloneqq C \setminus Z$ is minimized: $\costt{F,C'} \equiv \max_{j \in C'} \Big\{ \min_{i \in F} \big\{ d(i,j)^{\zl} \big\} \Big\}$.
\end{definition}

\noindent Similarly, we generalize the definition of constrained $k$-supplier problem to its outlier version, as follows:

\begin{definition}[Constrained Outlier $k$-Supplier Problem]\label{definition:constrained_outlier_k_supplier}
Let $(L,C,k,d,z,m)$ be any instance of the outlier $k$-supplier problem and $\mathbb{S}$ be any collection of partitionings such that any partitioning $\mathbb{O} \in \mathbb{S}$ is a partitioning of some subset $C' \subseteq C$ of size at least $|C|-m$. Find a clustering $\mathbb{O} = \{O_{1}, O_{2}, \dotsc, O_{k}\}$ in $\mathbb{S}$, that minimizes the objective function:
$\Psi^{*}(\OC) ~ \equiv ~ \min\limits_{\textrm{$k$-center-set $F$}} \Psi(F, \OC)$.
\end{definition}

\noindent Furthermore, we define the \emph{list outlier $k$-supplier problem} and \emph{outlier partition algorithm}, as follows:
\begin{definition}[List Outlier $k$-Supplier]
Let $\mathcal{I} = (L,C,k,d,z,m)$ be any instance of the outlier $k$-supplier problem and $\mathbb{S}$ be the collection of all possible $k$-partitionings of every subset $C' \subseteq C$ of size at least $|C|-m$.
Find a list $\mathcal{L}$ of $k$-center-sets (i.e., each element of the list is a set of $k$ elements from $L$) such that for any partitioning $\mathbb{O} \in \mathbb{S}$, the list $\mathcal{L}$ contains a $k$-center-set $F$ such that $\Psi(F,\OC) \leq \alpha \cdot OPT(\OC)$ for $\alpha = \rp{3}{z}$. For the $k$-center version $\alpha = \rp{2}{z}$.
\end{definition}

\begin{definition}[Outlier Partition Algorithm]
Let $\mathcal{I} = (L,C,k,d,z,m,\mathbb{S})$ be any instance of the constrained outlier $k$-supplier problem. Given a center set $F$, an outlier partition algorithm outputs a clustering in $\mathbb{S}$ that has the least clustering cost with respect to $F$.
\end{definition}

Now, suppose that we have an algorithm for the list outlier $k$-supplier problem and a partition algorithm for a particular constrained outlier $k$-supplier problem. Then, we can obtain an approximation algorithm for that constrained outlier $k$-supplier problem. The following theorem state this result and is analogous to Theorem~\ref{theorem:list_alpha_approx} for the non-outlier version.
\begin{theorem}\label{theorem:outlier_list_alpha_approx}
Let $\mathcal{I} = (L,C,k,d,z,m,\mathbb{S})$ be any instance of the constrained outlier $k$-supplier problem, and let $A_{\mathbb{S}}$ be any outlier partition algorithm for $\mathbb{S}$ with running time $T_{A}$.
Let $B$ be any algorithm for the list outlier $k$-supplier problem with running time $T_{B}$.
%Suppose we are given a list $\mathcal{L}$, then we can design 
Then, there is an algorithm that outputs a clustering $\OC \in \mathbb{S}$ that is an $\alpha$-approximate solution for the constrained outlier $k$-supplier instance $\mathcal{I}$. For the $k$-supplier objective, $\alpha = \rp{3}{z}$, and for the $k$-center objective, $\alpha = \rp{2}{z}$.
Moreover, the running time of the algorithm is $O(T_B + |\mathcal{L}| \cdot T_A)$.
\end{theorem}

\begin{proof}
The proof is analogous to Theorem~\ref{theorem:list_alpha_approx}.
\end{proof}

We design an $\mathsf{FPT}$ time algorithm for the list outlier $k$-supplier problem, parameterized by $k$ and $m$. Formally, we state the result as follows:
\begin{theorem}\label{theorem:list_Outlier_k_supplier}
Given an instance $\mathcal{I} = (L,C,k,d,z,m)$ of the outlier $k$-supplier problem. There is an algorithm for the list outlier $k$-supplier problem that outputs a list $\mathcal{L}$ of size at most $(k+m)^{O(k)} \cdot n$. Moreover, the running time of the algorithm is $(k+m)^{O(k)} \cdot n + O(n^2 \log n)$, which is $\mathsf{FPT}$ in $k$ and $m$.
\end{theorem}

\noindent Using Theorems~\ref{theorem:outlier_list_alpha_approx} and~\ref{theorem:list_Outlier_k_supplier}, we obtain the following two corollaries:
\begin{corollary} [Main Result]\label{corollary:outlier_1}
For any constrained version of the outlier $k$-supplier problem that has a partition algorithm with running time $T$, there exists a $\rp{3}{z}$ approximation algorithm with running time $T \cdot (k+m)^{O(k)} \cdot n + O(n^2 \log n)$.
\end{corollary}

\begin{corollary} [Main Result]\label{corollary:outlier_2}
For any constrained version of the outlier $k$-center problem that has a partition algorithm with running time $T$, there exists a $\rp{2}{z}$ approximation algorithm with running time $T \cdot (k+m)^{O(k)} \cdot n + O(n^2 \log n)$.
\end{corollary}

\noindent 
We consider the outlier versions of all the problems described in Table~\ref{table:definitions}. In Appendix~\ref{appendix:partition_algorithm}, we design $\mathsf{FPT}$ time partition algorithms for the outlier versions of all these problems. Thus, we get $\mathsf{FPT}$ time $\rp{3}{z}$ and $\rp{2}{z}$ approximation algorithm for these problems with respect to the outlier $k$-supplier and $k$-center objectives, respectively. Formally, we state the results as follows:

\begin{theorem}\label{theorem:outlier_1}
There is an $\mathsf{FPT}$ time $\rp{3}{z}$ approximation algorithm for the following constrained versions of the outlier $k$-supplier problems:
\begin{enumerate}
    \item $r$-gather outlier $k$-supplier problem
    \item $r$-capacity outlier $k$-supplier problem
    \item Balanced outlier $k$-supplier problem
    \item Chromatic outlier $k$-supplier problem
    \item Fault-tolerant outlier $k$-supplier problem
    \item Strongly private outlier $k$-supplier problem
\end{enumerate}
The running time of the algorithm is $(k+m)^{O(k)} \cdot n^{O(1)}$. Furthermore, for the $k$-center version, the approximation guarantee is $\rp{2}{z}$.
\end{theorem}

\noindent All the above mentioned problems generalize the unconstrained outlier $k$-supplier and $k$-center problems. Unsurprisingly, better approximation guarantees are not possible for these problems in $\mathsf{FPT}$ (in $k$ and $m$) time, assuming Gap-ETH. We give a formal proof of this result in Section~\ref{section:fpt_hardness_outlier}. Thus, it settles the complexity of the above mentioned $k$-supplier problems, parameterized by $k$ and $m$. For the $\ell$-diversity and fair outlier $k$-supplier problems, we obtain the following results:

\begin{theorem}\label{theorem:outlier_2}
There is $\mathsf{FPT}$ time $\rp{3}{z}$ and $\rp{2}{z}$ approximation algorithms for the $\ell$-diversity outlier $k$-supplier and $k$-center problems, respectively, with running time $(k \omega)^{O(k  \omega)} \cdot (k+m)^{O(k)} \cdot n^{O(1)}$.
\end{theorem}

\begin{theorem}\label{theorem:outlier_3}
There is $\mathsf{FPT}$ time $\rp{3}{z}$ and $\rp{2}{z}$ approximation algorithms for the fair outlier $k$-supplier and $k$-center problems, respectively, with running time $(k \Gamma)^{O(k  \Gamma)} \cdot (k+m)^{O(k)} \cdot n^{O(1)}$, where  $\Gamma$ denote the number of distinct collection of color classes induced by the colors of clients. Moreover, if the color classes are pair-wise disjoint, then $\Gamma = \omega$, and the running time of the algorithm is $(k \omega)^{O(k  \omega)} \cdot (k+m)^{O(k)} \cdot n^{O(1)}$.
\end{theorem}

\noindent The $\ell$-diversity outlier $k$-supplier problem for $m = 0$, $\omega = 1$, and $\ell = 1$, is equivalent to the unconstrained $k$-supplier problem. Also, the fair outlier $k$-supplier problem for $m = 0$, $\omega = 1$, $\beta_j = 0$, and $\alpha_j = 1$, is equivalent to the unconstrained $k$-supplier problem. Therefore, better approximation guarantees are not possible for these problems in $\mathsf{FPT}$ time parameterized by $m$, $k$, and $\Gamma$ (or $\omega$). The statement simply follows from Theorems~\ref{theorem:hardness_ksupplier} and~\ref{theorem:hardness_kcenter}.

Note that for $m = 0$, any constrained version of the outlier $k$-supplier problem simply corresponds to its non-outlier variant. Therefore, on substituting $m = 0$ in Theorems~\ref{theorem:outlier_1},~\ref{theorem:outlier_2}, and~\ref{theorem:outlier_3}, straightaway gives Theorems~\ref{theorem:1},~\ref{theorem:2}, and~\ref{theorem:3}, respectively.  Therefore, in Section~\ref{section:list_k_supplier}, we simply design an algorithm for the list outlier $k$-supplier problem, and in Appendix~\ref{appendix:partition_algorithm}, we design partition algorithms for the outlier versions of the constrained $k$-supplier problems. That would imply the results for their non-outlier counterparts too. In the next section, we discuss the known results for all the problems described in Table~\ref{table:definitions}.
\section{Related Work}\label{section:related_work}
In this section, we mention the best known results for different constrained clustering problems that we study in this paper. Furthermore, Table~\ref{table:known_results} nicely summarizes all these previously known results.

\begin{table}[htbp]
\begin{adjustbox}{width=\columnwidth,center}
\centering
\setcellgapes{1ex}\makegapedcells
\begin{tabular}{|l|l|l|l|l|l|}
\hline
\multirow{2}{*}{\#} & \multirow{2}{*}{\bf Problem} & \multicolumn{2}{l|}{\bf $k$-Supplier Objective} & \multicolumn{2}{l|}{\bf $k$-Center Objective} \\ \cline{3-6}
&  & Without Outliers ~ & With Outliers ~ & Without Outliers ~ & With Outliers ~ \\ \hline
\multirow{2}{*}{1.} & \makecell[l]{$r$-Gather Clustering Problem \\ (non-uniform version)} & \makecell[l]{-} & \makecell[l]{-} & \makecell[l]{-} & \makecell[l]{-}\\ \cline{2-6}
& \makecell[l]{$r$-Gather Clustering Problem \\ (uniform version)} & \makecell[l]{\textbf{3} ~\cite{rgather:Misc_2016_Ahmadian}\\ (poly time)} & \makecell[l]{\textbf{5}~\cite{rgather:Misc_2016_Ahmadian} \\ (poly time)} & \makecell[l]{\textbf{2}~\cite{rgather:k_center_2010_Aggarwal} \\(poly time)} & \makecell[l]{\textbf{4}~\cite{rgather:k_center_2010_Aggarwal}\\ (poly time)}\\ \hline
\multirow{2}{*}{2.} & \makecell[l]{$r$-Capacity  Clustering Problem\\ (non-uniform version)} & \makecell[l]{-} & \makecell[l]{-} & \makecell[l]{-} & \makecell[l]{-}\\ \cline{2-6}

& \makecell[l]{$r$-Capacity  Clustering Problem\\ (uniform version)} & \makecell[l]{\textbf{11}~\cite{capacitated:kcenter_2015_An} \\ (poly time)} & \makecell[l]{\textbf{13}~\cite{capacitated:kcenter_outliers_2014_Cygan_Tomasz} \\ (poly time)} & \makecell[l]{\textbf{5}~\cite{capacitated:kcenter_2000_khuller} \\ (poly time)} & \makecell[l]{\textbf{13}~\cite{capacitated:kcenter_outliers_2014_Cygan_Tomasz} \\ (poly time)} \\ \hline
\multirow{3}{*}{3.} & \makecell[l]{Balanced  Clustering Problem\\ (non-uniform version)} & \makecell[l]{-} & \makecell[l]{-} & \makecell[l]{-} & \makecell[l]{-}\\ \cline{2-6}
& \makecell[l]{Balanced  Clustering Problem\\ (uniform version)} &
\makecell[l]{\textbf{9}~\cite{capacitated:kcenter_rgather_outliers_2017_Hu_Ding_WADS} \\ (poly time)}& \makecell[l]{\textbf{13}~\cite{capacitated:kcenter_rgather_outliers_2017_Hu_Ding_WADS} \\ (poly time)}& \makecell[l]{\textbf{6}~\cite{capacitated:kcenter_rgather_outliers_2017_Hu_Ding_WADS} \\ (poly time) \vspace{1mm} \\  \textbf{4}~\cite{capacitated:kcenter_balanced_2017_CCCG_Hu_Ding} \\ ($\mathsf{FPT}$ time)} & \makecell[l]{\textbf{13}~\cite{capacitated:kcenter_rgather_outliers_2017_Hu_Ding_WADS} \\ (poly time)} \\ \hline
4. & \makecell[l]{Chromatic  Clustering Problem} & \makecell[l]{-} & \makecell[l]{-} & \makecell[l]{-} & \makecell[l]{-}\\ \hline
\multirow{2}{*}{5.} & \makecell[l]{Fault Tolerant  Clustering Problem\\ (non-uniform version)} & \makecell[l]{-}& \makecell[l]{-} & \makecell[l]{-} & \makecell[l]{-} \\ \cline{2-6}
& \makecell[l]{Fault Tolerant  Clustering Problem\\ (uniform version)} & \makecell[l]{\textbf{3}~\cite{fault:kcenter_2000_khuller} \\ (poly time)} & \makecell[l]{-} & \makecell[l]{\textbf{2}~\cite{fault:kcenter_2000_khuller} \\ (poly time)} & \makecell[l]{\textbf{6}~\cite{fault:outlier_kmedian_2020_Varadarajan} \\ (poly time)}  \\ \hline
%
% 8. & \makecell[l]{Balanced Fair $k$-supplier Problem\\ $k$-{\tt BFsupplier}} & \makecell[l]{-} & \makecell[l]{-} \\ \hline
% %
6. & \makecell[l]{Strongly Private  Clustering Problem} & \makecell[l]{\textbf{5}~\cite{rgather:2018_Rosner} \\ (poly time)} & \makecell[l]{-} & \makecell[l]{\textbf{4}~\cite{rgather:2018_Rosner} \\ (poly time)} & \makecell[l]{-} \\ \hline
7. & \makecell[l]{$\ell$-Diversity  Clustering Problem~\tablefootnote{Bercea~\emph{et al.}~\cite{fairness:2019_Bercea_Schmidt_APPROX} did not explicitly mention the results for the $\ell$-diversity clustering problem. However, the results follow from~\cite{fairness:2019_Bercea_Schmidt_APPROX} since the $\ell$-diversity problem is a special case of fair clustering problem with disjoint color classes as noted by Bandyapadhyay~\emph{et al.}~\cite{fairness:2021_constrained_Bandyapadhyay}.}} & \makecell[l]{\textbf{7}~\cite{fairness:2019_Bercea_Schmidt_APPROX} \\ (poly time)} & \makecell[l]{-} & \makecell[l]{\textbf{5}~\cite{fairness:2019_Bercea_Schmidt_APPROX} \\ (poly time)} & \makecell[l]{-} \\ \hline
\multirow{2}{*}{8.} & \makecell[l]{Fair Clustering Problem\\ (disjoint color classes)} & \makecell[l]{\textbf{7}~\cite{fairness:2019_Bercea_Schmidt_APPROX} \\ (poly time)} & \makecell[l]{-} & \makecell[l]{\textbf{5}~\cite{fairness:2019_Bercea_Schmidt_APPROX} \\ (poly time)} & \makecell[l]{-} \\ \cline{2-6}
& \makecell[l]{Fair Clustering Problem\\ (overlapping color classes)} & \makecell[l]{\boldmath{$(5,4\Delta+3)$} \\\cite{fairness:2019_Bera_NIPS} \\ (poly time)} & \makecell[l]{-} & \makecell[l]{\boldmath{$(3,4\Delta+3)$} \\\cite{fairness:2020_Harb_NIPS} \\ (poly time)} & \makecell[l]{-} \\ \hline
% 8. & \makecell[l]{Non-Uniform $k$-supplier problem \\ $k$-{\tt NUsupplier}} & \makecell[l]{-}
% & \makecell[l]{-} & \makecell[l]{-} & \makecell[l]{-}
% \\ \hline
%
% 9. & \makecell[l]{Outlier $k$-supplier problem \\ $k$-{\tt Osupplier}} & \makecell[l]{-}& \makecell[l]{-} & \makecell[l]{-} & \makecell[l]{-} \\ \hline
% % 
% 10. & \makecell[l]{Balanced Outlier $k$-supplier problem \\ $k$-{\tt BOsupplier}} & \makecell[l]{-} & \makecell[l]{-} & \makecell[l]{-} & \makecell[l]{-}
% \\ \hline
%
\end{tabular}
\end{adjustbox}
\caption{The known results for the constrained $k$-supplier/center problems with and without outliers for $z = 1$. We only mention the results for the \emph{soft assignment} version of the problems; however some of these results also hold for the \emph{hard assignment} version that we did not mention explicitly for the sake of simplicity. For the fair clustering problem with overlapping color classes, no polynomial time constant factor approximation algorithm is known yet. The algorithms of~\cite{fairness:2019_Bera_NIPS} and~\cite{fairness:2020_Harb_NIPS} give constant approximation guarantees; however, they violate the fair constraints by an additive factor of $4\Delta+3$, where $\Delta$ denotes the maximum number of color classes a client can be part of.}\label{table:known_results}
\end{table}

We would like to point out that the $r$-gather, $r$-capacity, and balanced $k$-supplier problems that we study in this paper impose cluster-wise constraints. However, the alternate definitions of the $r$-gather, $r$-capacity, and balanced $k$-supplier problems impose constraints on individual facility locations. 
Formally, the balanced $k$-supplier problem with location-wise constraints is defined as follows:
\begin{definition} [Balanced $k$-Supplier Problem with Location-Wise Constraints]\label{def:location_wise}
Given an instance $\mathcal{I} = (L,C,k,d,z)$ of the $k$-supplier problem, a lower bound function $g \colon L \to \mathbb{Z}_{+}$, and an upper bound function $h \colon L \to \mathbb{Z}_{+}$, find a set $F \subseteq L$ of $k$ facility and assignment $\phi \colon C \to L$ that minimizes the assignment cost $\max_{x \in C}  \big\{ d(x,\phi(x))^{\zl} \big\}$ and satisfies that $g(f) \leq |\phi^{-1}(f)| \leq h(f)$ for every facility location $f \in F$.  
\end{definition}

\noindent The above definition also encapsulates the $r$-gather and $r$-capacity $k$-supplier problems with location-wise constraints. For the $r$-gather problem, $h(f) = |C|$ for every facility location $f \in L$, and for the $r$-capacity problem, $g(f) = 0$ for every facility location $f \in L$. Moreover, when every facility location has the same values of $g(f)$ and $h(f)$, then the problems are known as the \emph{uniform $r$-gather, $r$-capacity, and balanced $k$-supplier problems}. It is easy to see that for the uniform version, the problem with location-wise constraints is equivalent to the problem with cluster-wise constraints. In other words, Definition~\ref{def:location_wise} is the same as the definition given in Table~\ref{table:definitions} for the uniform case. To the best of our knowledge, the non-uniform variant of the problem with cluster-wise constraints has not been studied before. Furthermore, we believe that it is non-trivial to obtain any polynomial time reduction between the problems with cluster-wise constraints and location-wise constraints.

% We show that the problem with cluster-wise constraints is easier than the problem with location-wise constraints. We describe the reduction in Appendix~\ref{appendix:reduction_balanced}.
% Therefore, the classical algorithms for the $r$-gather, $r$-capacity, and balanced clustering problems with location-wise constraints also holds for the cluster-wise constraints. 
Although Table~\ref{table:known_results} nicely summarizes all the previously known results, we defer the complete details to Appendix~\ref{appendix:related_work}. 
From Table~\ref{table:known_results}, we note that most of the known approximation algorithms have polynomial running time; however, they give much worse approximation guarantees than $3$ and $2$ for the $k$-supplier and $k$-center versions, respectively. 
Furthermore, the approximation guarantees worsen for the outlier version of these problems. Among the $\mathsf{FPT}$ time algorithms, the only known result is due to Hu Ding~\cite{rgather:balanced_CCCG_HU_Ding} for the uniform balanced $k$-center problem without outliers. The author used the classical $2$-approximation algorithm of Teofilo F. Gonzalez~\cite{kcenter:1985_Gonzalez} as a subroutine to obtain $\mathsf{FPT}$ time $4$-approximation algorithm for the problem. 
In this work, we improve this approximation guarantee to $2$ by using a bi-criteria approximation algorithm for the unconstrained $k$-center problem. 
In Section~\ref{section:list_k_supplier}, we describe the algorithm in detail. 
Also, note that our algorithm holds for both the $k$-supplier and $k$-center objectives with and without outliers. 
Moreover, our algorithm works for several other constrained clustering problems described in Table~\ref{table:definitions} and we expect this list to grow further when new interesting problems will be discovered. Moreover, these are the best approximation guarantees possible for these problems, parameterized by $k$ (and $m$ for the outlier version). 
In the next section, we define the notations that we use frequently in our proofs.

\section{Notations} \label{section:notations}
Let $\mathcal{I} = (L,C,k,d,z,m)$ denote any instance of the unconstrained outlier $k$-supplier problem. For the non-outlier version, $m = 0$. Let $F \subseteq L$ be any given center set. Let $C'$ be any subset of $C$. Then, we define the unconstrained $k$-supplier cost of $C'$ as:
\[
\costt{F,C'} \equiv \max_{x \in C'} \Big\{ d(F,x)^{\zl} \Big\}, \quad \textrm{where $d(F,x) = \min_{f \in F} \big\{ d(f,x)\big\}$ }
\]
If $F$ is a singleton set $\{f\}$, then we simply use the notation: $\costt{f,C'}$ instead of $\costt{\{f\},C'}$. Let $Z^{*}$ denote an optimal set of outliers and $F^{*}$ denote an optimal $k$-center set for the unconstrained outlier $k$-supplier instance $\mathcal{I}$. Then, we denote the optimal unconstrained outlier $k$-supplier cost of the instance by $OPT$. That is, $OPT \equiv \costt{F^{*},C \setminus Z^{*}}$. 

Let $\mathbb{O} = \{O_{1},\dotsc,O_k\}$ be any partitioning of a subset $C'$ of $C$. Let $F$ be a set of facility locations. Let $f_{i}^{*}$ be a facility in $F$ that minimizes the $1$-supplier cost of partition $O_i$. Then, the $k$-supplier cost of $\mathbb{O}$ with respect to the facility set $F$ is given as follows:

$$\Psi(F, \mathbb{O}) \equiv \max_{i = 1}^{k} \Big\{ \max_{x \in O_{i}} \big\{ d(x, f_{i}^{*})^{
\zl} \big\} \Big\}$$

\noindent In other words, a partition $O_i$ is completely assigned to the facility location $f_{i}^{*}$ in $F$, and the cost of every client in $O_i$ is measured with respect to $f_{i}^{*}$. Then, $\Psi(F,\mathbb{O})$ is simply the maximum assignment cost over all the clients. Furthermore, the optimal $k$-supplier cost of $\mathbb{O}$ is given as follows: $\Psi^{*}(\mathbb{O}) \equiv \min\limits_{\textrm{$k$-center-set $F$}} \Psi(F,\mathbb{O})$.

\section{Algorithm for List Outlier k-Supplier}\label{section:list_k_supplier}

In this section, we design an $\mathsf{FPT}$ (in $m$ and $k$) time algorithm for the list outlier $k$-supplier problem with running time $(k+m)^{O(k)} \cdot n^{O(1)}$. It implies a $k^{O(k)} \cdot n^{O(1)}$ time algorithm for the list $k$-supplier problem without outliers.
The algorithm is surprisingly simple. 
Let $\mathcal{I} = (L,C,k,d,z,m)$ be any instance of the outlier $k$-supplier problem. The algorithm consists of the following two parts:
\begin{enumerate}
    \item A $(1,O(\ln n))$ bi-criteria approximation algorithm for the outlier $k$-supplier problem. The algorithm outputs a set $S \subseteq L$ of $O(k \ln n)$ facilities and a set of $Z \subseteq C$ of $m$ outliers such that every client in $C \setminus Z$ has assignment cost at most the optimal cost $OPT$ of the outlier $k$-supplier instance $\mathcal{I}$. 
    \item For every outlier $x$ in $Z$, let $g(x)$ denote a facility in $L$ that is closest from $x$. Let $G = \{g(x) ~\mid ~ x \in Z \}$ be the set of such facilities. Let $C'$ be any arbitrary subset of $C$ of size at least $|C| - m$. We then show that for any arbitrary partitioning $\mathbb{O}$ of $C'$, there exists a $k$-sized subset $S' \subseteq S \cup G$ that gives $\rp{3}{z}$ approximation for $\mathbb{O}$. Therefore, we create a list $\mathcal{L}$ of all possible $k$-sized subsets of $S \cup G$. The list $\mathcal{L}$ is the required solution to the list outlier $k$-supplier problem. Moreover, the size of the list is $|\mathcal{L}| \leq O(m + k \ln n)^k = (k + m)^{O(k)} n$.
\end{enumerate}

\noindent We discuss the above two parts in more detail in Sections~\ref{subsection:1} and~\ref{subsection:2}, respectively.

\subsection{Bi-Criteria Approximation}\label{subsection:1}

In this subsection, we design a $(1,O(\ln n))$ bi-criteria approximation algorithm for the outlier $k$-supplier problem. An $(\alpha,\beta)$ bi-criteria approximation algorithm is defined as follows:

\begin{definition}[Bi-criteria Approximation]
Let $\mathcal{I} = (L,C,k,d,z,m)$ be any instance of the outlier $k$-supplier problem.
An $(\alpha,\beta)$ bi-criteria approximation algorithm is an algorithm that outputs a set $F' \subseteq L$ of $\beta k$ facilities and a set $Z' \subseteq C$ of at most $m$ ouliers such that the cost of the client set $C \setminus Z'$ with respect to $F'$ is at most $\alpha$ times the optimal cost of the instance. That is,
\[
\costt{F',C \setminus Z'} \leq \alpha \cdot \min_{\textrm{$k$-center-set $F$ and $|Z| \leq m$}} \Big\{ \costt{F,C \setminus Z} \Big\} = \alpha \cdot OPT
\]
\end{definition}

\noindent In the following lemma, we design an $(\alpha,\beta)$ bi-criteria approximation algorithm for the problem with $\alpha = 1$ and $\beta = O(\ln n)$.
\begin{lemma}\label{lemma:bicriteria}
Let $\mathcal{I} = (L,C,k,d,z,m)$ be any instance of the outlier $k$-supplier problem. Then, there exists a $(1,O(\ln n))$ bi-criteria approximation algorithm for the problem. Moreover, the running time of the algorithm is $O(n^2 \log n)$. 
\end{lemma}
\begin{proof}
We reduce the outlier $k$-supplier instance to a \emph{max $k$-coverage instance}. Some similar reductions have been previously carried out in the following works:~\cite{uncertain:kcenter_2008_Cormode,misc:1998_Asymmetric_Rina_Panigrahy}. A max $k$-coverage instance is denoted by $(U,\mathscr{C},k)$, where $U$ denotes the universal set and the set $\mathscr{C}$ is a collection of subsets of $U$. The task of the max $k$-coverage problem is to select at most $k$ sets from $\mathscr{C}$ whose union covers the maximum number of elements from $U$. For now, assume that we know the optimal cost $OPT$ of the outlier $k$-supplier instance. We then create a max $k$-coverage instance $(U,\mathscr{C},k)$, where $U$ corresponds to the client set $C$. In other words, for every client $j \in C$, there is an element $e_j$ in $U$. The set $\mathscr{C}$ corresponds to the facility location set $L$. That is, for every facility location $f \in L$, there is a set $S_{f} \in \mathscr{C}$. Furthermore, an element $e_j$ belongs to a set $S_{f}$ if and only if $d(j,f)^{\zl} \leq OPT$. 

Since there are $k$ facilities in $L$ that gives the optimal cost $OPT$ for at least $|C|-m$ clients in $C$, there exists $k$ sets in $\mathscr{C}$ that cover all but $m$ elements of $U$. For the max $k$-coverage problem, there exists a standard polynomial time greedy algorithm that selects $O(k \ln n)$ sets from $\mathscr{C}$ that cover at least as many elements of $U$ as covered by the optimal $k$ sets (see Section 35.3 of~\cite{book:CLRS_Third_Edition}). We use this greedy algorithm on $(U,\mathscr{C})$ to obtain a collection $\mathscr{C}' \subset \mathscr{C}$ of $O(k \ln n)$ sets that covers at least $|U|-m$ elements of $U$. Let $F'$ be the set of facility locations corresponding to $\mathscr{C}'$. Let $C'$ be the clients covered by $F'$. Then, every client in $C'$ has an  assignment cost of at most $OPT$ to one of the facilities in $F'$. Therefore, $F'$ is a $(1,O(\ln n))$ bi-criteria approximate solution to the problem and $Z' = C \setminus C'$ is the set of at most $m$ outlier points.

The only remaining task is to guess the value $OPT$ of the optimal solution. Since, there are at most $|L| \cdot |C|$ possible distances between a client and facility, we execute the greedy algorithm for every possibility. We choose the smallest distance for which the greedy algorithm outputs at most $O(k \ln n)$ sets and covers at least $|U|-m$ elements of $U$. The overall running time of the algorithm is polynomial in $n$. It can further be improved by performing a binary search on the $|L| \cdot |C|$ possible distances. More precisely, the reduction to the max $k$-coverage instance and the greedy algorithm both takes $O(n^2)$ time. After that $O(\log n)$ multiplicative factor due to binary search gives the overall time of $O(n^2 \log n)$ for the algorithm.  This completes the proof of the lemma.
\end{proof}

\noindent The above algorithm also applies to the outlier $k$-center problem since the outlier $k$-center problem is a special case of the outlier $k$-supplier problem. Next, we convert the bi-criteria approximation algorithm to a list outlier $k$-supplier/center algorithm. 
\subsection{Conversion: Bi-Criteria Approximation to List Outlier k-Supplier Algorithm}\label{subsection:2}

We prove the following lemma. 
\begin{lemma} \label{lemma:ksupplier_conversion}
Let $\mathcal{I} = (L,C,k,d,z,m)$ be any instance of the outlier $k$-supplier problem. Let $S$ be any $(1,O(\ln n))$ bi-criteria approximate solution of $\mathcal{I}$ and let $Z$ be the corresponding set of at most $m$ outliers. Let $g(x)$ denote a facility location in $L$ that is closest to an outlier $x \in Z$. 
Let $G = \{g(x) ~\mid ~ x \in Z \}$ be the set of such facilities. Let $C'$ be any subset of $C$ of size at least $|C| - m$. Then, for any arbitrary partitioning $\mathbb{O} = \{O_1,\dotsc,O_k\}$ of $C'$, there exists a $k$ sized subset $S'$ of $S \cup G$ that gives $\rp{3}{z}$ approximation for the cost of $\mathbb{O}$. That is,
\[
\Psi(S',\OC) \leq \rp{3}{z} \cdot \Psi^{*}(\OC).
\]
\end{lemma}
\begin{proof}
Suppose that for every client in $C$, the closest facility location in $S \cup G$ is given by a function $h \colon C \to S \cup G$. Let $F_{\mathbb{O}} = \{ f_{1},\dotsc,f_{k}\}$ denote the optimal facility set of $\mathbb{O}$ such that partition $O_{i}$ is assigned to center $f_i$. Let $x_{i}$ be any client in $O_{i}$. 
First, we observe that the cost of assigning $x_i$ to facility $h(x_{i})$ is at most $\Psi^{*}(\mathbb{O})$. That is, $d(x_i,h(x_i))^{\zl} \leq \Psi^{*}(\mathbb{O})$. To prove this statement, consider the case when $x_i \in Z$. Then $h(x_i) = g(x_i)$. In other words, $h(x_i)$ is the facility location in $L$ that is closest to $x_i$. Therefore, $d(x_i,h(x_i))^{\zl} \leq d(x_i,f_i)^{\zl} \leq \Psi^{*}(\mathbb{O})$. On the other hand, if $x_i \in C \setminus Z$, then the assignment cost $d(x_i,h(x_i))^{\zl}$ is at most $OPT$ since $S$ is a $(1,O(\ln n))$ bi-criteria approximate solution of $\mathcal{I}$. And, therefore, $d(x_i,h(x_i))^{\zl} \leq OPT \leq \Psi^{*}(\mathbb{O})$.

Now, we show that the set $S' \coloneqq \{h(x_1),\dotsc,h(x_k)\}$ is a $\rp{3}{z}$ approximation to the clustering cost of $\mathbb{O}$. That is, $\Psi(S',\mathbb{O}) \leq \rp{3}{z} \cdot \Psi^{*}(\mathbb{O})$. The proof follows from the following sequence of inequalities:
\begin{align*}
    \Psi(S',\mathbb{O}) &\leq \max_{i = 1}^{k} \Big\{ \costt{h(x_i),O_i} \Big\} \\
    &\leq  \max_{i = 1}^{k} \Big\{ \max_{x \in O_{i}} \big\{ d(h(x_i),x)^{\zl} \big\} \Big\} \\
    &\leq  \max_{i = 1}^{k} \Big\{ \max_{x \in O_{i}} \Big\{ \Big( d(x,f_i) + d(f_i,x_i) + d(x_i,h(x_i)) \Big)^{\zl} \Big\} \Big\} \\
    &\quad \quad \quad \quad \textrm{(using triangle inequality)}\\
    &\leq  \max_{i = 1}^{k} \Big\{ \max_{x \in O_{i}} \Big\{ \Big( d(x,f_i) + d(f_i,x_i) + \Psi^{*}(\mathbb{O})^{1/z} \Big)^{\zl} \Big\} \Big\} \\
    &\leq  \max_{i = 1}^{k} \Big\{ \max_{x \in O_{i}} \Big\{ \Big( 2 \cdot \Psi^{*}(\mathbb{O})^{1/z} + \Psi^{*}(\mathbb{O})^{1/z} \Big)^{\zl} \Big\} \Big\} \\
    &\leq  \max_{i = 1}^{k} \Big\{ \max_{x \in O_{i}} \Big\{ \Big( 3 \cdot \Psi^{*}(\mathbb{O})^{1/z} \Big)^{\zl} \Big\} \Big\} \\
    &= \rp{3}{z} \cdot \Psi^{*}(\mathbb{O})
\end{align*}
This completes the proof of the lemma.
\end{proof}

\noindent We show a similar result for the outlier $k$-center problem with improved approximation guarantee, as follows:

\begin{lemma} \label{lemma:kcenter_conversion}
Let $\mathcal{I} = (C,C,k,d,z,m)$ be any instance of the outlier $k$-center problem. Let $S$ be any $(1,O(\ln n))$ bi-criteria approximate solution of $\mathcal{I}$ and let $Z$ be the set of at most $m$ outliers. Let $C'$ be any subset of $C$ of size at least $|C| - m$. Then, for any arbitrary partitioning $\mathbb{O} = \{O_1,\dotsc,O_k\}$ of $C'$, there exists a $k$ sized subset $S'$ of $S \cup Z$ that gives $\rp{2}{z}$ approximation for the cost of $\mathbb{O}$. That is,
\[
\Psi(S',\OC) \leq \rp{2}{z} \cdot \Psi^{*}(\OC)
\]
\end{lemma}
\begin{proof}
Suppose that for every client in $C$, the closest facility location in $S \cup Z$ is given by a function $h \colon C \to S \cup Z$. Let $F_{\mathbb{O}} = \{ f_{1},\dotsc,f_{k}\} \subseteq C$ denote the optimal facility set of $\mathbb{O}$ such that partition $O_{i}$ is assigned to facility $f_i$. Note that $f_i$ is also a client location. Now, observe that the cost of assigning a client $f_i$ to facility location $h(f_{i})$ is at most $\Psi^{*}(\mathbb{O})$. That is, $d(f_i,h(f_i))^{\zl} \leq \Psi^{*}(\mathbb{O})$. To prove this statement, consider the case when $f_i \in Z$, then $h(f_i) = f_i$. Therefore, $d(f_i,h(f_i))^{\zl} = 0 \leq \Psi^{*}(\mathbb{O})$. On the other hand, if $f_i \in C \setminus Z$, then the assignment cost $d(f_i,h(f_i))^{\zl}$ is at most the optimal cost $OPT$ since $S$ is a $(1,O(\ln n))$ bi-criteria approximate solution of $\mathcal{I}$. And, therefore, $d(f_i,h(f_i))^{\zl} \leq OPT \leq \Psi^{*}(\mathbb{O})$.

Now, we show that the set $S' \coloneqq \{h(f_1),\dotsc,h(f_k)\}$ is a $\rp{2}{z}$ approximation to the clustering cost of $\mathbb{O}$. That is, $\Psi(S',\mathbb{O}) \leq \rp{2}{z} \cdot \Psi^{*}(\mathbb{O})$. The proof follows from the following sequence of inequalities:
\begin{align*}
    \Psi(S',\mathbb{O}) &\leq \max_{i = 1}^{k} \Big\{ \costt{h(f_i),O_i} \Big\} \\
    &\leq  \max_{i = 1}^{k} \Big\{ \max_{x \in O_{i}} \big\{ d(h(f_i),x)^{\zl} \big\} \Big\} \\
    &\leq  \max_{i = 1}^{k} \Big\{ \max_{x \in O_{i}} \Big\{ \Big( d(x,f_i) + d(f_i,h(f_i)) \Big)^{\zl} \Big\} \Big\},\quad  \textrm{(using triangle inequality)}\\
    &\leq  \max_{i = 1}^{k} \Big\{ \max_{x \in O_{i}} \Big\{ \Big( d(x,f_i) + \Psi^{*}(\mathbb{O})^{1/z} \Big)^{\zl} \Big\} \Big\} \\
    &\leq  \max_{i = 1}^{k} \Big\{ \max_{x \in O_{i}} \Big\{ \Big( \Psi^{*}(\mathbb{O})^{1/z} + \Psi^{*}(\mathbb{O})^{1/z} \Big)^{\zl} \Big\} \Big\} \\
    &= \rp{2}{z} \cdot \Psi^{*}(\mathbb{O})
\end{align*}
This completes the proof of the lemma.
\end{proof}

\noindent The above two lemmas gives the following corollary for the  list outlier $k$-supplier/center problem.

\begin{corollary}[Main Result]
Given any instance $\mathcal{I} = (L,C,k,d,z,m)$ of the outlier $k$-supplier problem. There is an algorithm that outputs a list $\mathcal{L}$ of $k$-center-sets such that for any arbitrary partitioning $\mathbb{O}$ of any subset $C'$ of $C$ of size at least $|C|-m$, there is a center set $F$ in the list that gives $\rp{3}{z}$ approximation for the clustering cost of $\mathbb{O}$. That is,
\[
\Psi(F,\mathbb{O}) \leq \alpha \cdot \Psi^{*}(\mathbb{O}), \quad \textrm{for $\alpha = \rp{3}{z}$}
\]
For the outlier $k$-center instance $\mathcal{I} = (C,C,k,d,z,m)$, the approximation factor is $\alpha = \rp{2}{z}$. Moreover, the size of the list is at most $(k+m)^{O(k)} \cdot n$ and the running time of the algorithm is $(k+m)^{O(k)} \cdot n + O(n^2 \log n)$, which is $\mathsf{FPT}$ in $k$ and $m$.
\end{corollary}
\begin{proof}
By Lemma~\ref{lemma:ksupplier_conversion}, we have that there is a $k$ sized subset $S' \subseteq S \cup G$ that gives $\rp{3}{z}$ approximation for the clustering $\mathbb{O}$. The size of the set $S \cup G$ is at most $m + O(k \ln n)$. We create the list $\mathcal{L}$ by adding to it all possible $k$ sized subsets of $S \cup G$. Therefore, $|\mathcal{L}| = (m+ O(k \ln n))^k$. If $m \leq k \ln n$, then $|\mathcal{L}| = k^{O(k)} \cdot (\ln n)^k$. Further, using the inequality that $(\ln n)^k \leq k^{O(k)} \cdot n$, we get $|\mathcal{L}| \leq k^{O(k)} \cdot n$. On the other hand, if $m > k \ln n$, then $|\mathcal{L}| \leq m^{O(k)}$. This proves that $|\mathcal{L}| \leq (m+k)^{O(k)} \cdot n$ for the list outlier $k$-supplier problem. 

Moreover, the sets $S$ and $Z$ can be computed in time $O(n^2 \log n)$ using Lemma~\ref{lemma:bicriteria}. And, the set $G$ can be computed from the set $Z$ in time $O(m \cdot n) = O(n^2)$. Thus, the overall running time of the algorithm is $(k+m)^{O(k)} \cdot n + O(n^2 \log n)$. 
A similar proof for the $k$-center version follows from Lemma~\ref{lemma:kcenter_conversion}. This proves the corollary.
\end{proof}
\section{FPT Hardness: Outlier k-Supplier Problem}\label{section:fpt_hardness_outlier}
In this section, we establish $\mathsf{FPT}$ hardness of approximation results for the outlier $k$-supplier and $k$-center problems. The proofs are rather trivial. For the outlier $k$-supplier problem, we obtain the following result:
\begin{theorem}\label{theorem:hardness_outlier_ksupplier}
For any constant $\veps>0$, $z > 0$, and any function $g \colon \mathbb{Z}^{+} \times \mathbb{Z}^{\geq 0} \to \mathbb{R}^{+}$, the outlier $k$-supplier problem can not be approximated to factor $(\rp{3}{z} - \veps)$ in time $g(k,m) \cdot n^{m+o(k)} $, assuming Gap-ETH.
\end{theorem}
\begin{proof}
For the sake of contradiction, assume that for some constant $\veps>0$, $z > 0$, and function $g \colon \mathbb{Z}^{+} \times \mathbb{Z}^{\geq 0} \to \mathbb{R}^{+}$ the outlier $k$-supplier problem can be approximated to factor $(\rp{3}{z} - \veps)$ in time $g(k,m) \cdot n^{m+o(k)}$. For $m = 0$, the problem is the classical (non-outlier) $k$-supplier problem. Therefore, the $k$-supplier problem can be approximated to factor $(\rp{3}{z} - \veps)$ in time $g(k,0) \cdot n^{o(k)} = h(k)\cdot n^{o(k)}$ for some function $h \colon \mathbb{Z}^{+} \to \mathbb{R}^{+}$. It contradicts Theorem~\ref{theorem:hardness_ksupplier}. Thus, it proves the theorem.
\end{proof}

\noindent Similarly, we obtain the following $\mathsf{FPT}$ hardness of approximation result for the outlier $k$-center problem:
\begin{theorem}\label{theorem:hardness_outlier_kcenter}
For any constant $\veps>0$, $z > 0$, and any function $g \colon \mathbb{Z}^{+} \times \mathbb{Z}^{\geq 0} \to \mathbb{R}^{+}$, the outlier $k$-center problem can not be approximated to factor $(\rp{2}{z} - \veps)$ in time $g(k,m) \cdot n^{m+o(k)}$, assuming Gap-ETH.
\end{theorem}
\begin{proof}
For the sake of contradiction, assume that for some constant $\veps > 0$, $z > 0$, and function $g \colon \mathbb{Z}^{+} \times \mathbb{Z}^{\geq 0} \to \mathbb{R}^{+}$ the outlier $k$-center problem can be approximated to factor $(\rp{2}{z} - \veps)$ in time $g(k,m) \cdot n^{m+o(k)}$. For $m = 0$, the problem is the classical (non-outlier) $k$-center problem. Therefore, the $k$-center problem can be approximated to factor $(\rp{2}{z} - \veps)$ in time $g(k,0) \cdot n^{o(k)} = h(k)\cdot n^{o(k)}$ for some function $h \colon \mathbb{Z}^{+} \to \mathbb{R}^{+}$. It contradicts Theorem~\ref{theorem:hardness_kcenter}. Thus, it proves the theorem.
\end{proof}

\noindent Note that the above two results do not eliminate the possibility of having the polynomial time $\rp{3}{z}$ and $\rp{2}{z}$ approximation algorithms for the outlier $k$-supplier and $k$-center problems, respectively. In fact, there exists polynomial time $\rp{3}{z}$ and $\rp{2}{z}$ approximation algorithms for the outlier $k$-supplier~\cite{outlier:kcenter_2001_Charikar} and $k$-center problem~\cite{outlier:kcenter_2020_Deeparnab}, respectively. Moreover, we can even obtain the optimal solutions to the outlier $k$-supplier and $k$-center problems using a trivial $O(n^{k+1}k)$ running time algorithm.
\section{Conclusion}\label{section:conclusion}
In this paper, we worked within the unified framework of Ding and Xu~\cite{constrained:2015_Ding_and_Xu} to obtain $\mathsf{FPT}$ time approximation algorithms for a range of constrained $k$-supplier/center problems in general metric spaces. Surprisingly, the algorithm turned out to be simple via a natural reduction to the maximum $k$-coverage problem.
Moreover, we noted that better approximation guarantees are not possible for the constrained $k$-supplier/center problems in $\mathsf{FPT}$ time parameterized by the number of clusters $k$, assuming Gap-ETH. This settles the complexity of the problem as far as parmeterization by $k$ is concerned. However, our algorithm only works for the soft assignment version of the constrained clustering problems. It will be interesting to extend the algorithm for the hard assignment version.
Furthermore, it will be interesting to explore other constrained problems that can fit into the unified framework and can benefit from the results in this work.

\addcontentsline{toc}{section}{References}
\bibliographystyle{alpha}
\bibliography{references}

\appendix

\section{FPT Hardness of Approximation}\label{appendix:reduction}

In this section, we prove the $\mathsf{FPT}$ hardness of approximation result for the $k$-center problem using a reduction from the set coverage problem.
\begin{definition}[Set Coverage]
Given an integer $k>0$, a set $U$, and a collection $\mathscr{C} = \{S_1,\dotsc,S_m\}$ of subsets of $U$, i.e., $S_j \subseteq U$ for every $j \in [m]$, determine if there exist $k$ sets in $\mathscr{C}$ that cover all elements in $U$.
\end{definition}
%\ragesh{A definition ending in a question did not seem right. So, I changed it to "determine if"...}

\noindent
\textbf{Reduction to $k$-Center Problem:}  Given a set coverage instance $(U,\mathscr{C},k)$, we construct a $k$-center instance $(C,C,k,d,z)$ as follows. Let $C$ be partitioned into two sets $C_{\ell}$ and $C_{r}$. For every set $S_{i} \in \mathscr{C}$, we define a location $f_{i} \in C_{\ell}$. For every element $e \in U$, we define a location $x_{e} \in C_{r}$. Now, let us define the distance function $d(.,.)$ as follows. For any two clients $x_{e},x_{e'} \in C_{r}$ the distance $d(x_{e},x_{e'}) = 2$. For any two clients $f_{i},f_{j} \in C_{\ell}$ the distance $d(f_{i},f_{j}) = 1$. For any client $f_{i} \in C_{\ell}$ and $x_{e} \in C_{r}$, if $e \notin S_{i}$, the distance $d(x_e,f_{i}) = 2$; otherwise $d(x_e,f_{i}) = 1$. Furthermore, assume that $d(.,.)$ is a symmetric function, i.e., $d(x,y) = d(y,x)$ for every $x,y \in C$. Also assume that $d(x,x) \geq 0$ for every $x \in C$. It is easy to see that $d(.,.)$ satisfies all the properties of a metric space. This completes the construction.

In a YES instance of the set coverage problem, there exists $k$ sets: $S_{i_{1}},\dotsc,S_{i_{k}}$ in $\mathscr{C}$ that covers all elements of $U$, i.e., $S_{i_{1}} \cup \dotsc \cup S_{i_{k}} = U$. Then, the center set $F = \{f_{i_{1}},\dotsc,f_{i_{k}}\} \subseteq C_{\ell}$ gives the $k$-center cost $1$. In a NO instance of the set coverage problem, there does not exist any $k$ sets in $\mathscr{C}$ that cover all elements of $U$. Then, we show that for any center set $F \subseteq C$ of size $k$, the $k$-center cost is $\rp{2}{z}$. We prove this statement using contraposition. Suppose there is a center set $F \subseteq C$ that gives $k$-center cost $1$. Then, any facility $f$ in $F \cap C_{r}$ can not serve a client in $C_{r}$ other than $f$ itself; otherwise the $k$-center cost would be $\rp{2}{z}$. Therefore, we move all the facilities in $F \cap C_{r}$ to their closest locations in $C_{\ell}$; the $k$-center cost stays at most $1$. It means, there are $k$ facilities in $F \cap C_{\ell}$ that give $k$-center cost $1$. In other words, there are $k$ sets in $\mathscr{C}$ that cover all elements of $U$. This proves the contrapositive statement.  
Therefore, for a NO instance of the set coverage problem, the $k$-supplier cost of the reduced instance is $\rp{2}{z}$. Since the set coverage problem is $\mathsf{W[2]}$-hard~\cite{Vertex_Cover:2013_Set_Cover_W2_Hard}, it implies that the $k$-supplier problem can not be approximated to any factor better than $\rp{2}{z}$, in $\mathsf{FPT}$ time, assuming $\mathsf{FPT} \neq \mathsf{W[2]}$. Moreover, the following stronger $\mathsf{FPT}$ hardness result is known for the set coverage problem (see Corollary 2 and Theorem 25 of~\cite{fpt:2020_Pasin_SODA}):

\begin{theorem}
For any function $g \colon \mathbb{Z}^{+} \to \mathbb{R}^{+}$, there is no $g(k) \cdot n^{o(k)}$ time algorithm for the set coverage problem, assuming Gap-ETH.
\end{theorem}

\noindent This implies the following $\mathsf{FPT}$ hardness of approximation for the $k$-center problem.

\begin{theorem}\label{theorem:hardness_supplier2}
For any constant $\veps>0$, $z >0$, and any function $g \colon \mathbb{Z}^{+} \to \mathbb{R}^{+}$, the $k$-center problem can not be approximated to factor $(\rp{2}{z} - \veps)$ in time $g(k) \cdot n^{o(k)}$, assuming Gap-ETH.
\end{theorem}
\section{Proof of Theorem~\ref{theorem:list_alpha_approx}}\label{appendix:proof_theorem_3}

\begin{theorem}\label{theorem:list_alpha_approx_2}Let $\mathcal{I} = (L,C,k,d,z,\mathbb{S})$ be any instance of the constrained $k$-supplier problem, and let $A_{\mathbb{S}}$ be a partition algorithm for $\mathbb{S}$ with running time $T_{A}$.
Let $B$ be any algorithm for the list $k$-supplier problem with running time $T_{B}$.
Then, there is an algorithm that outputs a clustering $\OC \in \mathbb{S}$ that is an $\alpha$-approximate solution for the constrained $k$-supplier instance $\mathcal{I}$. For the $k$-supplier objective, $\alpha = \rp{3}{z}$, and for the $k$-center objective, $\alpha = \rp{2}{z}$. Moreover, the running time of the algorithm is $O(T_B + |\mathcal{L}| \cdot T_A)$.
\end{theorem}
\begin{proof}
The algorithm is simple. We first run algorithm $B$ to obtain a list $\mathcal{L}$.
For every $k$-center-set in the list, the algorithm runs the partition algorithm $A_{\mathbb{S}}$ on it. 
Then the algorithm outputs a center set that gives the minimum clustering cost. 
Let $F'$ be this $k$-center-set and $\OC'$ be the corresponding clustering. 
We claim that $(F',\OC')$ is an $\alpha$-approximation for the constrained $k$-supplier problem.

Let $\OC^{*}$ be an optimal solution for the constrained $k$-supplier instance $(L,C,k,d,z,\mathbb{S})$ and $F^{*}$ denote the corresponding $k$-center-set. 
By the definition of the list $k$-supplier problem there is a $k$-center-set $F$ in the list $\mathcal{L}$, such that $\Psi(F,\OC^{*}) \leq \alpha \cdot \Psi(F^{*},\OC^{*})$. 
Let $\OC = A_{\mathbb{S}}(F) \in \mathbb{S}$ be the optimal clustering corresponding to $F$. 
Thus, $\Psi(F,\OC) \leq \alpha \cdot \Psi(F^{*},\OC^{*})$. 
Since $(F',\mathbb{O}')$ gives the minimum cost clustering in the list, we have $\Psi(F',\OC') \leq \Psi(F,\OC)$. 
Therefore, $\Psi(F',\OC') \leq \alpha \cdot \Psi(F^{*},\OC^{*})$. 

The running time analysis is also simple. $T_{B}$ is the time to obtain the list $\mathcal{L}$. Then, the algorithm runs the partition procedure $A_{\mathbb{S}}$ for every center set in the list, the running time of this step is $|\mathcal{L}| \cdot T_A$. 
Picking a minimum cost clustering from the list takes $O(|\mathcal{L}|)$ time. Hence the overall running time is $O(T_B + |\mathcal{L}| \cdot T_A)$.
\end{proof}
\section{Fault-Tolerant k-Supplier Problem}\label{appendix:fault_tolerant}
In this section, we show a reduction from the fault-tolerant $k$-supplier problem to the chromatic $k$-supplier problem. In the fault-tolerant $k$-supplier problem, there is no explicit constraint imposed on any cluster. However, the cost function is different from the classical $k$-supplier problem. That is, for every client $x \in C$, we are given an integer $\ell_{x} \leq k$. Then, for a set $F = \{f_1,\dotsc,f_k\}$ of facility locations, the assignment cost of a client $x$ is proportional to the distance to its $\ell_{x}^{th}$ closest facility location in $F$. Thus, the overall cost of a fault-tolerant instance is given as: $\ftcostt{F,C} \equiv \max_{x \in C} \big\{ d'(F,x)^{z} \big\}$, where $d'(F,x)$ is the distance of $x$ to the $\ell_{x}^{th}$ closest facility location in $F$. Therefore, the problem does not satisfy the constrained clustering framework. More precisely, the problem does not satisfy Definition~\ref{definition:constrained_k_supplier} of the constrained $k$-supplier problem. 
We reduce any instance of the fault-tolerant $k$-supplier problem to an instance of the chromatic $k$-supplier problem so that it satisfies the constrained clustering framework. The reduction is the same as the one used by Ding and Xu~\cite{constrained:2015_Ding_and_Xu} for the $k$-median/means objective. Let $\mathcal{I} = (L,C,k,d,z,S_{x})$ be any instance of the fault-tolerant $k$-supplier problem, where $S_{x} = \{\ell_{x} ~ \mid ~ x \in C \}$ is a set of integers for the clients in $C$. Firstly, we color each client in $C$ with different color. This coloring is given by a bijective function $f \colon C \to R$, where $R = \{r_1,\dotsc,r_{|C|}\}$ is a set of $|C|$ different colors. Then, we create $\ell_{x}$ copies of each client $x$ and assign each copy the same color as $x$. Let $C'$ be this new client set and color on a client is denoted by function $g \colon C' \to R$. This completes the construction of the chromatic $k$-supplier instance. Let this new instance be $\mathcal{I}' = (L,C',k,d,z,g)$. Now, we show the following lemma:

\begin{lemma}
Let $\mathcal{I} = (L,C,k,d,z,S_{x})$ be any instance of the fault-tolerant $k$-supplier problem and $\mathcal{I}' = (L,C',k,d,z,g)$ be the corresponding instance of the chromatic $k$-supplier problems as defined above. For any set $F = \{f_1,\dotsc,f_k\}$ of facility locations, the fault-tolerant $k$-supplier cost of $C$ with respect to $F$ is the same as the chromatic $k$-supplier cost of $C'$ with respect to $F$.  
\end{lemma}
\begin{proof}
The fault-tolerant $k$-supplier cost of $C$ with respect to $F$ is $\max_{x \in C} \big\{ d'(F,x)^{z}\big\}$, where $d'(F,x)$ is the distance of $x$ to its $\ell_{x}^{th}$ closest facility location in $F$. Now, let us evaluate the chromatic $k$-supplier cost of $C'$. For a client $x \in C$, let $\{x_1,\dotsc,x_{\ell_{x}}\}$ denote the copies of $x$ in $C'$. These copies share the same color.  
By the definition of the chromatic $k$-supplier problem, a cluster can not contain two clients of the same color. In other words, $\{x_1,\dotsc,x_{\ell_{x}}\}$ must get assigned to different facility locations in $F$. Since all these clients are co-located, the assignment cost is minimized when each of them is assigned to one of the $\ell_{x}$ closest facility locations in $F$. Without loss of generality, let $x_{\ell_{x}}$ is the client that is assigned to $\ell_{x}^{th}$ closest facility location in $F$. Note that, $x_{\ell_{x}}$ has the maximum assignment cost among $\{x_1,\dotsc,x_{\ell_{x}}\}$. This cost is the same as the assignment cost of $x \in C$ in the fault-tolerant $k$-supplier instance.
Therefore, the overall chromatic $k$-supplier cost of the clients in $C'$ is: $\max_{x \in C} \big\{ d'(F,x)^{z} \big\}$. This cost is the same as the fault-tolerant $k$-supplier cost of $C$ with respect to $F$. This completes the proof of the lemma.
\end{proof}

\noindent From the above lemma, we can say that any $\alpha$-approximate solution to the reduced chromatic $k$-supplier instance is also an $\alpha$-approximate solution to the original fault-tolerant $k$-supplier instance, and vice-versa. 
\section{Partition Algorithm}\label{appendix:partition_algorithm}

In this section, we design $\mathsf{FPT}$ time partition algorithms for all the problems mentioned in Table~\ref{table:definitions} along with their outlier version. For the first six problems in the table, we define a new hybrid problem that encapsulates all the constraints of these problems. Therefore, instead of designing a partition algorithm for each problem separately, we design a partition algorithm for the hybrid problem. Formally, the problem is defined as follows:

\begin{definition}[Hybrid $k$-Supplier]
Given an instance $\mathcal{I} = (L,C,k,d,z,m)$ of the outlier $k$-supplier problem, a partitioning of the client set $C$ into $\omega$ color classes: $C_1,\dotsc,C_{\omega}$, vectors: $\ell,r \in \mathbb{Z}^{k}$ and $\alpha,\beta \in \mathbb{Z}^{\omega}$, find a subset $Z \subseteq C$ of at most $m$ outliers and a partitioning $\mathbb{O} = \{O_1,\dotsc,O_k\}$ of the set $C \setminus Z$ with minimum $\Psi^{*}(\mathbb{O})$ such that it satisfies that $\ell_{i} \leq |O_i| \leq r_i$ and $\alpha_j \leq |O_i \cap C_j| \leq \beta_j$ for every $i \in [k]$ and $j \in [\omega]$.
\end{definition}

\noindent The first six problems given in Table~\ref{table:definitions} along with their outlier versions are the special cases of the hybrid $k$-supplier problem. Furthermore, the hybrid $k$-supplier problem satisfies the definition of the constrained outlier $k$-supplier problem, i.e., Definition~\ref{definition:constrained_outlier_k_supplier}. For the sake of completeness, we describe how the hybrid $k$-supplier problem encapsulates the first six problems in Table~\ref{table:definitions}:

\begin{enumerate}
    \item For every $i \in [k]$ and $j \in [\omega]$, if $r_{i} = |C|$, $\alpha_j = 0$, and $\beta_j = |C|$, then the hybrid $k$-supplier problem is equivalent to the $r$-gather outlier $k$-supplier problem.
    \item For every $i \in [k]$ and $j \in [\omega]$, if $\ell_{i} = 0$, $\alpha_j = 0$, and $\beta_j = |C|$, then the hybrid $k$-supplier problem is equivalent to the $r$-capacity outlier $k$-supplier problem.
    \item For every $j \in [\omega]$, if $\alpha_j = 0$ and $\beta_j = |C|$, then the hybrid $k$-supplier problem is equivalent to the balanced outlier $k$-supplier problem.
    \item For every $i \in [k]$ and $j \in [\omega]$, if $\ell_i = 0$, $r_i = |C|$, $\alpha_j = 0$, and $\beta_j = 1$, then the hybrid $k$-supplier problem is equivalent to the chromatic outlier $k$-supplier problem.
    \item The fault-tolerant $k$-supplier problem can be reduced to the chromatic $k$-supplier problem as described in Appendix~\ref{appendix:fault_tolerant}. Therefore, the fault tolerant outlier $k$-supplier problem is also a special case of the hybrid $k$-supplier problem.
    \item For every $i \in [k]$ and $j \in [\omega]$, if $\ell_i = 0$, $r_i = |C|$, and $\beta_j = |C|$, then the hybrid $k$-supplier problem is equivalent to the strongly private outlier $k$-supplier problem.
\end{enumerate}

\noindent Now, we give a partition algorithm for the hybrid $k$-supplier problem.

% \begin{lemma}
% Given an instance $\mathcal{I} = (L,C,k,d,z,m,C_1,\dotsc,C_{\omega},\ell,r,\alpha,\beta)$ of the hybrid $k$-supplier problem and a facility set $F = \{f_1,\dotsc,f_k\} \subseteq L$. There is an algorithm that outputs a subset $Z \subseteq C$ of at most $m$ outliers and a partitioning $\mathbb{O} = \{O_1,\dotsc,O_{k}\}$ of $C \setminus Z$ that minimizes $\Psi(F,\mathbb{O})$ and satisfies the hybrid $k$-supplier constraints. Moreover, the running time of the algorithm is $k^{O(k)} \cdot n^{O(1)}$, which is $\mathsf{FPT}$ in $k$.
% \end{lemma}
\begin{lemma}
There is a $k^{O(k)} \cdot n^{O(1)}$ time partition algorithm for the hybrid $k$-supplier problem.
\end{lemma}

\begin{proof}
Let $F$ be a given set of $k$ facility locations. Let $Z^{*}$ be an optimal outlier set, and let $\mathbb{O}^{*} = \{O_1^{*},\dotsc,O_k^{*}\}$ be an optimal partitioning of $C \setminus Z^{*}$ that minimizes the objective function $\Psi(F,\mathbb{O})$. 
Every partition $O_i^{*}$ is assigned to one of the facility location in $F$. Since, the algorithm does not know $O_{i}^{*}$, it makes a guess that $O_{i}^{*}$ is assigned to a facility $f \in F$. Since there are $k$ facilities in $F$, there are total $k^{k}$ possibilities of assigning every $O_i$ to facilities in $F$. The algorithm also make a guess on the value of the optimal solution $\Psi(F,\mathbb{O}^{*})$ over $|L| \cdot |C|$ possible distances between clients in $C$ and facility locations in $L$. For each particular guess, the algorithm checks if there is some feasible assignment of clients to $F$ satisfying the hybrid constraints. Then, the algorithm outputs the feasible assignment that has the minimum assignment cost over all the possible guesses. Now, suppose that in a particular guess, $O_{i}^{*}$ is assigned to a facility $f_{i} \in F$ and the optimal cost is $\lambda$. The algorithm reduces the assignment problem to a circulation problem on a flow network $G = (V,E)$. The flow network is shown in Figure~\ref{figure:flow}. \\
\begin{figure}[h]
\begin{framed}
\includegraphics[width=12cm, height=6cm]{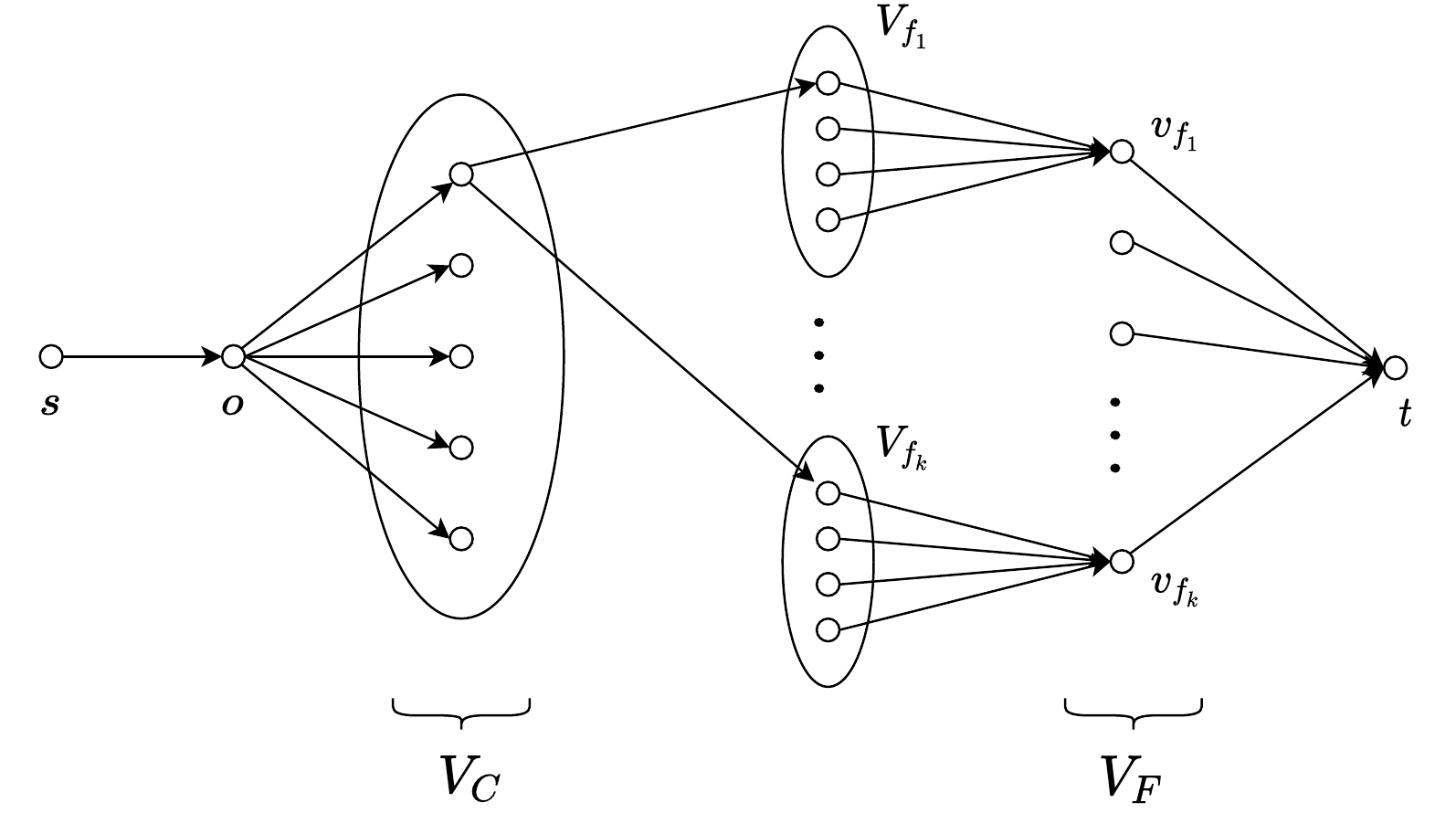}
\centering
\caption{The flow network $G = (V,E)$ that is used in the partition algorithm of the hybrid $k$-supplier problem.}
\label{figure:flow}
\end{framed}
\vspace*{-4mm}
\end{figure}

\noindent The vertex set $V$ is partitioned into the following sets: 
\begin{enumerate}
    \item A source vertex $s$, sink vertex $t$, and vertex $o$.
    \item A vertex set $V_{C}$ that corresponds to the client set $C$. In other words, for every client $x \in C$ there is a vertex $v_x$ in $V_{C}$.
    \item A vertex set $V_{F}$ that corresponds to the facility set $F$. In other words, for every facility $f_i \in F$ there is a vertex $v_{f_i} \in V_{F}$.
    \item The vertex sets: $V_{f_1},\dotsc,V_{f_k}$. A vertex set $V_{f_{i}}$ is defined as $\{ v_{f_i,1},\dotsc,v_{f_i,\omega}\}$ such that a vertex $v_{f_{i},j}$ corresponds to a facility $f_i \in F$ and a color class $j \in [\omega]$.
\end{enumerate}

\noindent Furthermore, the edge set $E$ is defined with the lower and upper bound flow constraints as follows:

\begin{enumerate}
    \item There is a directed edge $(s,o)$ with lower bound flow of $|C|-m$ and upper bound flow of $|C|$. This edge ensures that at least $|C|-m$ clients are assigned to $F$.
    \item For every vertex $v_x \in V_C$, there is a directed edge $(o,v_x)$ with upper bound flow of $1$ and lower bound flow of $0$. These edges ensure that a client is assigned at most one facility in $F$.
    \item For every vertex $v_x \in V_C$ and vertex $v_{f_{i},j} \in V_{f_{i}}$, there is a directed edge $(v_x,v_{f_{i},j})$ if and only if the client $x$ belongs to color class $C_j$ and $d(x,f_{i})^{\zl} \leq \lambda$. These edges have the upper bound flow of $1$ and lower bound flow of $0$. These edges ensures a feasible assignment of cost at most $\lambda$.
    \item For every $i \in [k]$ and $j \in [\omega]$, there is a directed edge $(v_{f_{i},j},v_{f_{i}})$ with lower bound flow of $\alpha_{j}$ and upper bound flow of $\beta_{j}$. These edges ensures that for every color class $j \in [\omega]$ and every facility $f_i \in F$, the constraint $\alpha_{j} \leq |O_{i}^{*} \cap C_{j}| \leq \beta_{j}$ is satisfied.
    \item For every $i \in [k]$, there is a directed edge $(v_{f_{i}},t)$ with lower bound flow $\ell_{i}$ and upper bound flow $r_{i}$. These edges ensures that for every partition ${O}_{i}^{*} \in \{O_{1}^{*},\dotsc,O_{k}^{*}\}$, the constraint $\ell_{i} \leq |O_{i}^{*}| \leq r_{i}$ is satisfied. 
\end{enumerate}

\noindent It is easy to see that a feasible integral flow through the network $G$ corresponds to an assignment of at least $|C|-m$ clients in $C$ to $F$ that satisfies the hybrid constraints. Moreover, the assignment cost is at most $\lambda$. The circulation problem on $G$ can be solved in polynomial time using the algorithms in~\cite{circulation:1972_Edmond_Karp,circulation:1985_Eva_Tardos,circulation:2012_Nikhil_Bansal}. Since we run this algorithm for $k^{k} \cdot |C| \cdot |L|$ possible guesses, the overall running time of the partition algorithm is $k^{k} \cdot n^{O(1)}$. This completes the proof of the lemma. 
\end{proof}

\noindent In the next subsection, we design the partition algorithms for the remaining two problems in Table~\ref{table:definitions}.

\subsection{Partition Algorithm: $\ell$-Diversity and Fair Outlier $k$-Supplier Problems}\label{appendix:fair_partition}
Bandyapadhyay~\emph{et al.}~\cite{fairness:2021_constrained_Bandyapadhyay} noted that the $\ell$-diversity clustering problem is a special case of the fair clustering problem. In other words, if the fair clustering problem has disjoint color classes, and $\alpha_{j} = 1/\ell$ and $\beta_j = 0$ for every color class $C_j \in \{C_1,\dotsc,C_{\omega}\}$, then the problem is equivalent to the $\ell$-diversity clustering problem. Therefore, it is sufficient to design a partition algorithm for the fair outlier $k$-supplier problem. 

For the fair \emph{$k$-median} problem (without outliers), Bandyapadhyay~\emph{et al.}~\cite{fairness:2021_constrained_Bandyapadhyay} designed an $\mathsf{FPT}$ time partition algorithm. We simply extend their algorithm to the fair \emph{$k$-supplier} problem with outliers. 
Let $\mathcal{I} = (L,C,k,d,z,m,C_1,\dotsc,C_{\omega},\alpha,\beta)$ be any instance of the fair outlier $k$-supplier problem. 
Recall that the sets $C_1,\dotsc,C_{\omega}$ are the color classes, each of which is a subset of the client set $C$. Moreover, any two color classes can overlap with each other. In other words, a client in $C$ might belong to different color classes. To simplify the problem, we partition the client set into $\Gamma$ disjoint groups: $P_1,\dotsc,P_{\Gamma}$ such that the points belonging to the same group $P_i$ belong to the same set of colored classes. In other words, if clients $x$ and $y$ belong to the same group $P_i$, then for every color class $C_t \in \{C_1,\dotsc,C_{\omega}\}$, $x \in C_t$ if and only if $y \in C_t$. Also, if $x$ and $y$ belong to different groups $P_i$ and $P_j$, respectively, then there exists a color class $C_{t} \in \{C_1,\dotsc,C_{\omega}\}$ such that $x \in C_t$ and $y \notin C_t$, or $x \notin C_t$ and $y \in C_t$. Note that if the color classes are pair-wise disjoint, then $\Gamma$ is equal to the number of color classes, i.e., $\Gamma = \omega$.
Now, we design $\mathsf{FPT}$ time partition algorithm for the fair outlier $k$-supplier problem with running time $(k \Gamma)^{k \Gamma} \cdot n^{O(1)}$, where $\Gamma$ is the number of distinct collection of color classes induced by the colors of clients. Formally, we state the result as follows:

\begin{lemma}
For the fair outlier $k$-supplier problem, there is a $(k \Gamma)^{k \Gamma} \cdot n^{O(1)}$ time partition algorithm, where $\Gamma$ is the number of distinct collection of color classes induced by the colors of clients.
\end{lemma}
\begin{proof}

Let $F = \{f_1,\dotsc,f_k\}$ is the given set of facility locations for which we want an optimal partitioning satisfying the fair outlier constraints. Let $Z$ denote the optimal set of outliers and $OPT$ denote the optimal $k$-supplier cost of assigning $C \setminus Z$ to $F$ while satisfying the fair constraints. Since there are $|F| \cdot |C|$ possible distances between $C$ and $F$, the algorithm tries each possibility for $OPT$. For each possibility, the algorithm finds a feasible assignment of clients to facilities. Then, the algorithm outputs that assignment that gives the minimum assignment cost. The assignment problem can modeled as an integer linear program; however solving an integer linear program is $\mathsf{NP}$-hard in general. 
Therefore, we model the problem as a mixed integer linear program, which can be solved in $\mathsf{FPT}$ time parameterized by the number of integer variables. The following is a formal statement for the same:
\begin{prop}[Proposition 8.1 of~\cite{fairness:2021_constrained_Bandyapadhyay}]\label{prop:MILP}
Given a real valued matrix $A \in \mathbb{R}^{m \times d}$, vector $b \in \mathbb{R}^{m}$, vector $c \in \mathbb{R}^{d}$, and a positive integer $p \leq d$. There is an $\mathsf{FPT}$ time algorithm that finds a vector $x = (x_1,\dotsc,x_{d}) \in \mathbb{R}^{d}$ that minimizes $c \cdot x$, and satisfies that $A \cdot x \leq b$ and $x_1,\dotsc,x_p \in \mathbb{Z}$. The running time of the algorithm is $O(p^{2.5p+o(p)}d^{4}B)$ and the space complexity is polynomial in $B$, where $B$ is the bitsize of the given instance.
\end{prop}

\noindent Now, we model the assignment problem as a mixed integer linear program (MILP), as follows: 
\begin{framed}
\vspace{-2.5mm}
\begin{align*}
    \textrm{Constraints:} \quad \quad &\sum_{f \in F} g_{x,f} \leq 1 \quad &&\textrm{for every client $x \in C$} \\
    &\sum_{x \in P_i} g_{x,f} = h_{f,i} \quad &&\textrm{for every group $P_i \in \{P_1,\dotsc,P_{\Gamma}\}$ and facility $f \in F$} \\
    &\sum_{f \in F, x \in C} g_{x,f} \geq |C| - m \quad &&\textrm{for client set $C$ and facility set $F$} \\
    & \sum_{x \in C_{j}} g_{x,f} \leq \alpha_{j} \cdot \sum_{x \in C} g_{x,f}   &&\textrm{for every facility $f \in F$ and color class $C_{j} \in \{C_1,\dotsc,C_{\omega} \}$} \\
    & \sum_{x \in C_{j}} g_{x,f} \geq \beta_{j} \cdot \sum_{x \in C} g_{x,f}   &&\textrm{for every facility $f \in F$ and color class $C_{j} \in \{C_1,\dotsc,C_{\omega} \}$} \\
    &0 \leq g_{x,f} \leq 1 \quad &&\textrm{for every client $x \in C$ and facility $f \in F$} \\
    & h_{f,i} \in \mathbb{Z}_{\geq 0} \quad &&\textrm{for every group $P_i \in \{P_1,\dotsc,P_{\Gamma}\}$ and facility $f \in F$} \\
\end{align*}
\vspace{-12.0mm}
\end{framed}

\noindent In the above MILP, for every client $x \in C$ and facility $f \in F$, there is a fractional variable $g_{x,f} \leq 1$ that denotes the fraction of client $x$ assigned to facility $f$. Moreover, for every group $P_{i} \in \{P_1,\dotsc,P_{\Gamma}\}$ and facility $f \in F$, there is an integer variable $h_{f,i}$ that denotes the total fraction of clients in $P_i$ that is assigned to facility $f$. In other words, $h_{f,i} = \sum_{x \in P_{i}} g_{x,f}$. The third constraint of the MILP corresponds to the number of outliers being at most $m$. Lastly, the fourth and fifth constraints of MILP correspond to the fairness constraints. 

We solve the above mixed integer linear program in time $O(k \Gamma)^{k \Gamma} \cdot n^{O(1)}$ as per Proposition~\ref{prop:MILP}. However, the obtained optimal solution contains fractional $g_{x,f}$ values.Next, we show that there always exists a solution with integral $g_{x,f}$ values that can be obtained in polynomial time. We reduce the problem to the \emph{circulation} problem on directed graphs. We construct a flow network $G = (V,E)$ with upper and lower bound flow requirements on every edge. The construction is shown in Figure~\ref{figure:partition_fair}. 

\begin{figure}[h]
\begin{framed}
\includegraphics[width=12cm, height=6cm]{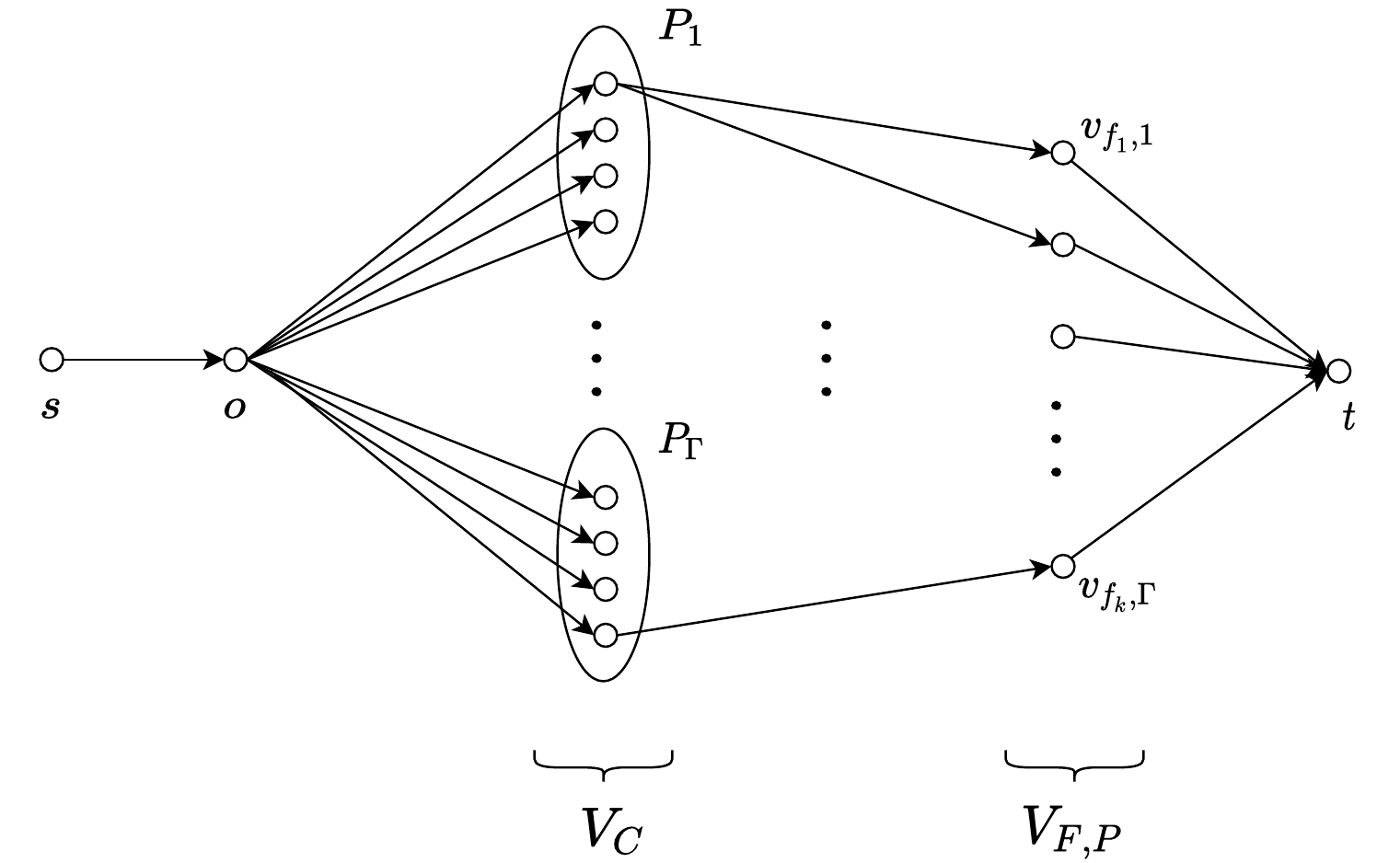}
\centering
\caption{ The flow network $G = (V,E)$ that is used by the partition algorithm of the fair outlier $k$-supplier problem.}
\label{figure:partition_fair}
\end{framed}
\vspace*{-4mm}
\end{figure}

The vertex set $V$ contains a source vertex $s$, a sink vertex $t$, and an \emph{outlier regulating vertex} $o$. We will describe the functioning of $o$ shortly.
The rest of the vertex set is partitioned into two sets: $V_{C}$ and $V_{F,P}$. The vertex set $V_{C}$ corresponds to the client set $C$. In other words, for every client $x \in C$ there is a vertex $v_{x} \in C$. The vertex set $V_{F,P}$ corresponds to facility set $F$ and groups $P_1,\dotsc,P_{\Gamma}$. In other words, for every facility $f \in F$ and group $P_i \in \{P_1,\dotsc,P_{\Gamma}\}$, there is a vertex $v_{f,i}$ in $V_{F,P}$. Furthermore, the edge set $E$ is partitioned as follows. There is an edge $(s,o)$ with lower and upper bound flow requirement of exactly $\sum_{x \in C, f \in F} g_{x,f}$. Note that $\sum_{x \in C, f \in F} g_{x,f}$ is an integer since $\sum_{x \in C, f \in F} g_{x,f} = \sum_{i = 1}^{\Gamma} \sum_{f \in F} h_{i,f}$ and $h_{i,f}$'s are integers. Also note that $\sum_{x \in C, f \in F} g_{x,f} \geq |C|-m$ as per Constraint $2$ of the MILP. Therefore, it ensures that at least $|C|-m$ clients are assigned to $F$, and that's why we call the vertex $o$, the outlier regulating vertex. Furthermore, the edge set $E$ contains for every vertex $v_{x} \in C$, an edge $(o,v_{x})$ with lower bound flow requirement $0$ and upper bound flow requirement $1$. These edges ensure that a client is assigned to at most one facility in $F$. Furthermore, for every vertex $v_{x} \in V_{C}$ and $v_{f,i} \in V_{F,P}$, there is an edge $(v_{x},v_{f,i})$ if and only if $x \in P_{i}$. Each edge has lower bound flow requirement $0$ and upper bound flow requirement $1$. Lastly, for every vertex $v_{f,i} \in V_{F,P}$, there is an edge $(v_{f,i},t)$ with lower and upper bound flow requirement of exactly $h_{f,i}$. These edges ensure that exactly $h_{f,i}$ clients of $P_{i}$ are assigned to any facility $f \in F$. It is easy to see that the constructed flow network $G$ admit a feasible flow if we send a flow of $g_{x,f}$ through every edge $(v_{x},v_{f,i})$. Moreover, this is the maximum flow that we can send through the network since the capacity of edge $(s,o)$ is $\sum_{x \in C, f \in F} g_{x,f}$. Since the lower and upper bound requirements on every edge is an integer, there exists a feasible integral flow through the network. We find that flow in polynomial time using the algorithms in~\cite{circulation:1985_Eva_Tardos,circulation:1972_Edmond_Karp,circulation:2012_Nikhil_Bansal}. Let the new integral flow through $(v_x,v_{f,i})$ is $g'_{x,f}$. Then, we show that the values $g'_{x,f}$'s also satisfy the MILP constraints. The constraints $(1)$-$(3)$ of the MILP are trivially satisfied from the flow network. Also note that $\sum_{x \in P_i} g'_{x,f} = h_{f,i}$ due to the flow network. Therefore, the fair constraints $(4)$ and $(5)$ are satisfied since $\sum_{x \in C} g_{x,f} = \sum_{x \in C} g'_{x,f}$ and $\sum_{x \in C_j} g_{x,f} = \sum_{x \in C_j} g'_{x,f}$ for every facility $f \in F$ and every color class $C_{j} \in \{C_1,\dotsc,C_{\omega}\}$. These two equalities follow from the following two sequence of equalities:
\begin{enumerate}
    \item  $$\sum_{x \in C} g_{x,f} = \sum_{i = 1}^{\Gamma}\sum_{x \in P_{i}} g_{x,f} = \sum_{i = 1}^{\Gamma} h_{f,i} = \sum_{x \in C} g'_{x,f}$$  
    \item $$\sum_{x \in C_j} g_{x,f} = \sum_{i \colon P_{i} \subseteq C_j }\sum_{x \in P_{i}} g_{x,f} = \sum_{i \colon P_{i} \subseteq C_j } h_{f,i} = \sum_{x \in C_j} g'_{x,f}$$
\end{enumerate}

\noindent This completes the proof of the lemma.
\end{proof}

\section{Related Work in Detail}\label{appendix:related_work}
In this section, we discuss the known results for the constrained $k$-supplier problems in detail. A summary of these results in given in Table~\ref{table:known_results}.

\begin{enumerate}
    \item \textbf{$r$-gather $k$-supplier problem:} For the $r$-gather $k$-supplier problem  with and without outliers, Ahmadian and Swamy~\cite{rgather:Misc_2016_Ahmadian} gave $5$ and $3$ approximation algorithms, respectively (see Theorems $15$ and $16$ of~\cite{rgather:Misc_2016_Ahmadian}). The running time of their algorithm is polynomial in the input size and it also holds for the non-uniform variant of the problem with location-wise constraints. For the $r$-gather $k$-center problem (with uniform lower bounds)  with and without outliers, Aggarwal~\emph{at al.}~\cite{rgather:k_center_2010_Aggarwal} gave $4$ and $2$ approximation algorithms, respectively (see Section 2.4 of~\cite{rgather:k_center_2010_Aggarwal}). The running time of their algorithm is polynomial in the input size. 
    \item \textbf{$r$-capacity $k$-supplier problem:} For the $r$-capacity $k$-supplier and $k$-center problems (without outliers), An~\emph{et al.}~\cite{capacitated:kcenter_2015_An} gave $11$ and $9$ approximation algorithms, respectively (see Theorems 1 and 6 of~\cite{capacitated:kcenter_2015_An}). Their algorithm is polynomial time and also works for the non-uniform variant of the problem with location-wise constraints. For the $r$-capacity $k$-supplier problem (without outliers) with uniform upper bounds, Khuller and Sussmann~\cite{capacitated:kcenter_2000_khuller} gave $5$ approximation algorithm.\footnote{The authors call the soft assignment version of the $r$-capacity $k$-center problem as the \emph{capacitated multi-$k$-center problem}.} For the outlier version, Cygan and Kociumaka~\cite{capacitated:kcenter_outliers_2014_Cygan_Tomasz} gave a 25-approximation algorithm for the $r$-gather $k$-supplier and $k$-center problems with non-uniform upper bounds imposed on facility locations. Moreover, the authors improved the approximation guarantee to $13$ for the uniform $r$-gather $k$-supplier and $k$-center problems with outliers (see Theorem 1 and Corollary 1 of Theorem 1 of~\cite{capacitated:kcenter_outliers_2014_Cygan_Tomasz}). All the above mentioned algorithms have polynomial running time.
    \item \textbf{Balanced $k$-supplier problem:} For the balanced $k$-supplier/center problem with non-uniform lower and upper bounds, no constant factor approximation is known yet. However, for the uniform lower bounds and non-uniform upper bounds on locations (without outliers), Ding~\emph{et al.}~\cite{capacitated:kcenter_rgather_outliers_2017_Hu_Ding_WADS} gave 13 and 9 approximation algorithms for the balanced $k$-supplier and $k$-center problems, respectively. Furthermore, for the same variant, the authors gave $25$ approximation algorithm for both the balanced $k$-supplier and $k$-center problems with outliers. For the uniform lower bounds and uniform upper bounds (without outliers), the authors improved the approximation bounds to $9$ and $6$ for the balanced $k$-supplier and $k$-center problems, respectively. Furthermore, for the same variant, the authors gave $13$ approximation algorithm for both the balanced $k$-supplier and $k$-center problems with outliers. All the above mentioned algorithms have polynomial running time.
    
    For the balanced $k$-center problem (without outliers) with uniform lower and upper bounds, Hu Ding~\cite{rgather:balanced_CCCG_HU_Ding} gave a $4$-approximation algorithm with $\mathsf{FPT}$ running time $k^{O(k)} \cdot n^{O(1)}$. In this work, we improve this approximation guarantee to $2$.
    Furthermore, we give $\mathsf{FPT}$ time $3$-approximation algorithm for the balanced $k$-supplier problem.
    \item \textbf{Chromatic $k$-supplier problem:} For the chromatic $k$-supplier problem, no constant factor approximation is known yet. However, for the $k$-median and $k$-means objectives, $\mathsf{FPT}$ time approximation algorithms are known in the continuous Euclidean space and general discrete metric spaces~\cite{constrained:2015_Ding_and_Xu,chromatic:k_cones_ICALP_2011_ding_and_xu,fairness:2021_constrained_Bandyapadhyay,constrained:2020_GJK_FPT}. 
    \item \textbf{Fault-tolerant $k$-supplier problem:} In the fault-tolerant $k$-supplier problem, given a facility set $F \subseteq L$, the cost of a client $x \in C$ is proportional to the distance to its $\ell_{x}$ closest facility location in $F$. If $\ell_{x}$ is the same for every $x$ in $C$, then we call the problem the \emph{uniform fault-tolerant $k$-supplier problem}. On the other hand, if $\ell_{x}$ is not the same for every $x$, then we call the problem the non-uniform fault-tolerant $k$-supplier problem. For the non-uniform version, no constant factor approximation is known yet. However, for the uniform fault-tolerant $k$-supplier and $k$-center problem (without outliers), Khuller~\emph{et al.}~\cite{fault:kcenter_2000_khuller} gave $3$ and $2$ approximation algorithms, respectively. Recently, Inamdar and Varadarajan gave $6$ approximation algorithm for the uniform fault-tolerant \emph{$k$-center} problem with outliers. However, no constant factor approximation is known for the uniform fault-tolerant \emph{$k$-supplier} problem with outliers. 
    \item \textbf{Strongly private $k$-supplier problem:} This problem has recently been proposed by Rösner and Schmidt~\cite{rgather:2018_Rosner}. The authors gave $5$ and $4$ approximation algorithms for the $k$-supplier and $k$-center versions, respectively, without outliers. No constant factor approximation algorithm is known for the problem with outliers. 
    \item \textbf{$\ell$-diversity $k$-supplier problem:} Bandyapadhyay~\emph{et al.}~\cite{fairness:2021_constrained_Bandyapadhyay} noted that the $\ell$-diversity clustering problem is a special case of the fair clustering problem when the color classes are disjoint, and $\alpha_{j} = 1/\ell$ and $\beta_j = 0$ for every color class $C_j \in \{C_1,\dotsc,C_{\omega}\}$.  Therefore, the problem admit $7$ and $5$ approximation algorithms for the $k$-supplier and $k$-center versions, respectively, without outliers~\cite{fairness:2019_Bercea_Schmidt_APPROX}. 
    
    There is another variant of the $\ell$-diversity clustering problem that was proposed by Li~\emph{et al.}~\cite{L_diversity:2010_kcenter_Li_Jian}. Given an integer constant $\ell \geq 0$, the task is to find a clustering $\OC = \{O_1, ..., O_t\}$ of the client set $C$ with minimum $\Psi^{*}(\OC)$ such that for every cluster $O_i$, $|O_i|\geq \ell$ and $O_i$ should not have any two clients with the same color. Note that unlike other clustering problems, here we do not have any restriction on the number of open centers. The authors gave $2$-approximation algorithm for the problem for $L = C$. When the number of partitions are restricted to $k$, then the problem has not been studied before. However, the problem satisfies the Definition~\ref{definition:constrained_k_supplier} of the constrained clustering problem. Furthermore, the problem admit an $\mathsf{FPT}$ time partition algorithm since it is a special case of the hybrid $k$-supplier problem discussed in Appendix~\ref{appendix:partition_algorithm}. Therefore, the problem has $\mathsf{FPT}$ (in $k$) time $3$ and $2$ approximation algorithms for the $k$-supplier and $k$-center versions, respectively. 
    \item \textbf{Fair $k$-supplier problem:} In this problem, we are given $\omega$ color classes: $C_1,\dotsc,C_{\omega}$ that are subsets of the client set $C$. When the color classes are pair-wise disjoint, $7$ and $5$ approximation algorithms are known for the $k$-supplier and $k$-center versions, respectively, without outliers~\cite{fairness:2019_Bercea_Schmidt_APPROX}. On the other hand, when the color classes are not necessarily disjoint, no true constant factor approximation algorithm is known yet. The existing algorithms give $5$ and $3$ approximation guarantees for the $k$-supplier and $k$-center objectives, respectively; however, they violate the fairness constraint by an additive factor of $4\Delta +3$, where $\Delta$ denote the maximum number of groups a client can be part of~\cite{fairness:2019_Bera_NIPS,fairness:2020_Harb_NIPS}. Moreover, no constant factor approximation algorithm is known for the outlier-version of this problem.  
    
\end{enumerate}

\end{document}